\documentclass[sigconf, nonacm, pdfa]{acmart}

\newcommand\vldbdoi{XX.XX/XXX.XX}
\newcommand\vldbpages{XXX-XXX}
\newcommand\vldbvolume{19}
\newcommand\vldbissue{5}
\newcommand\vldbyear{2026}
\newcommand\vldbauthors{\authors}
\newcommand\vldbtitle{\shorttitle} 
\newcommand\vldbavailabilityurl{https://github.com/NTU-Siqiang-Group/ArceKV}
 
\usepackage{subcaption}

\usepackage{MnSymbol}
\usepackage{algpseudocode}
\usepackage[ruled,linesnumbered,vlined]{algorithm2e}

\usepackage{algorithmicx}
\usepackage{textgreek}
\usepackage{tablefootnote}
\usepackage{multirow}
\usepackage{overpic}
\usepackage{enumitem}
\usepackage{bm}
\usepackage{xcolor}
\usepackage{tabu}
\usepackage{amsmath}
\usepackage{tablefootnote}
\usepackage{subcaption}
\usepackage{tcolorbox}
\usepackage{caption}
\usepackage{subcaption}
\usepackage[dvipsnames]{xcolor}
\usepackage{bbding}
\usepackage{enumitem}
\usepackage{wrapfig}
\usepackage{xpatch}

\usepackage[normalem]{ulem}
\useunder{\uline}{\ul}{}

\newcommand\oureng{\textsc{Arce}}
\newcommand\ourkv{\textsc{ArceKV}}
\newcommand\ourlsm{\textsc{ElasticLSM}}

\definecolor{ao(english)}{rgb}{0.0, 0.5, 0.0}
\newcounter{dingheng}
\numberwithin{dingheng}{section}

\newcounter{fan}
\numberwithin{fan}{section}

\newcounter{junfeng}
\numberwithin{junfeng}{section}

    \newcounter{Siqiang}
\numberwithin{Siqiang}{section}

\setcopyright{none}

\usepackage[linewidth=1pt]{mdframed}
\let\maketitlesup\maketitle\vspace{3ex}
\xpatchcmd{\maketitlesup}{\@mkteasers}{}{}{}
\xpatchcmd{\maketitlesup}{\@mkabstract}{}{}{}

\begin{document}
\title
[Response Letter for `{\ourkv}: Towards Workload-driven LSM-compactions for Key-Value Store Under Dynamic Workloads'']
{Response Letter for PVLDB Submission 812 \\
"{\ourkv}: Towards Workload-driven LSM-compactions for Key-Value Store Under Dynamic Workloads"}

\settopmatter{printfolios=false}

\setlength{\textfloatsep}{6.0pt plus 0.8pt minus 5.6pt}
\setlength{\floatsep}{6.0pt plus 0.8pt minus 5.2pt}
\setlength{\intextsep}{6.0pt plus 0.8pt minus 5.8pt}
\setlength{\abovedisplayskip}{4.pt plus 0.4pt minus 3.0pt}
\setlength{\belowdisplayskip}{4.pt plus 0.4pt minus 3.0pt}
\tcbset{colframe=white!75!black, boxrule=0.3mm, left=0.4mm, right=0.4mm, top=0.4mm, bottom=0.4mm, arc=0.75mm, after skip=1mm, before skip=2mm}
\begin{abstract}
Key-value stores underpin a wide range of applications due to their simplicity and efficiency. Log-Structured Merge Trees (LSM-trees) dominate as their underlying structure, excelling at handling rapidly growing data. Recent research has focused on optimizing LSM-tree performance under static workloads with fixed read-write ratios.
However, real-world workloads are highly dynamic, and existing workload-aware approaches often struggle to sustain optimal performance or incur substantial transition overhead when workload patterns shift. To address this, we propose {\ourlsm}, which removes traditional LSM-tree structural constraints to allow more flexible management actions (i.e., compactions and write stalls) creating greater opportunities for continuous performance optimization. We further design {\oureng}, a lightweight compaction decision engine that guides {\ourlsm} in selecting the optimal action from its expanded action space.
Building on these components, we implement {\ourkv}, a full-fledged key-value store atop RocksDB. Extensive evaluations demonstrate that {\ourkv} outperforms state-of-the-art compaction strategies across diverse workloads, delivering around 3$\times$ faster performance in dynamic scenarios.
\end{abstract}

\title{{\ourkv}: Towards Workload-driven LSM-compactions for Key-Value Store Under Dynamic Workloads}
\author{Junfeng Liu}
\email{junfeng001@e.ntu.edu.sg}
\affiliation{%
  \institution{Nanyang Technological University}
  \country{Singapore}
}
\author{Haoxuan Xie}
\email{haoxuan001@e.ntu.edu.sg}
\affiliation{%
  \institution{Nanyang Technological University}
  \country{Singapore}
}

\author{Siqiang Luo}
\email{siqiang.luo@ntu.edu.sg}
\affiliation{%
  \institution{Nanyang Technological University}
  \country{Singapore}
}


 
\maketitle
\begingroup\small\noindent\raggedright\textbf{PVLDB Reference Format:}\\
\vldbauthors. \vldbtitle. PVLDB, \vldbvolume(\vldbissue): \vldbpages, \vldbyear.\\
\href{https://doi.org/\vldbdoi}{doi:\vldbdoi}
\endgroup
\begingroup
\renewcommand\thefootnote{}\footnote{\noindent
This work is licensed under the Creative Commons BY-NC-ND 4.0 International License. Visit \url{https://creativecommons.org/licenses/by-nc-nd/4.0/} to view a copy of this license. For any use beyond those covered by this license, obtain permission by emailing \href{mailto:info@vldb.org}{info@vldb.org}. Copyright is held by the owner/author(s). Publication rights licensed to the VLDB Endowment. \\
\raggedright Proceedings of the VLDB Endowment, Vol. \vldbvolume, No. \vldbissue\ %
ISSN 2150-8097. \\
\href{https://doi.org/\vldbdoi}{doi:\vldbdoi} \\
}\addtocounter{footnote}{-1}\endgroup

\ifdefempty{\vldbavailabilityurl}{}{
\vspace{.3cm}
\begingroup\small\noindent\raggedright\textbf{PVLDB Artifact Availability:}\\
The source code, data, and/or other artifacts have been made available at \url{\vldbavailabilityurl}.
\endgroup
}

\setcounter{figure}{0}
\setcounter{table}{0}
\setcounter{section}{0}
\setcounter{equation}{0}
\setcounter{page}{1} 

\section{Introduction}
\label{sec:intro}
\begin{table*}[t]
\caption{Comparison between {\ourlsm} and existing workload-aware LSM-tree structures. The example assumes a three-level LSM-tree and a MemTable size of $F$.}
\vspace{-3mm}
\renewcommand\arraystretch{1.28}
\scalebox{0.78}{
\begin{tabular}{@{}|l|cccc|cccc|@{}}
\toprule
\multicolumn{1}{|c|}{\multirow{2}{*}{Methods}} & \multicolumn{4}{c|}{Structural Configuration}                                                                                                                                                             & \multicolumn{4}{c|}{LSM Management Actions}                                                                                                                                                                                                                                                                                                           \\ \cmidrule(l){2-9} 
\multicolumn{1}{|c|}{}                         & \multicolumn{1}{c|}{LSM structure}               & \multicolumn{1}{c|}{Size Ratios}                  & \multicolumn{1}{c|}{Level Capacities}                    & \#Sorted Run                            & \multicolumn{1}{c|}{Trigger Compaction}                                           & \multicolumn{1}{c|}{Picked Runs}                                                                          & \multicolumn{1}{c|}{Write stall}                                                  & Dynamic Workloads                                                 \\ \midrule
\multirow{2}{*}{Dostoevsky}                    & \multicolumn{1}{c|}{\multirow{2}{*}{Fluid Tree}} & \multicolumn{1}{c|}{\multirow{2}{*}{$\{T,T,T\}$}} & \multicolumn{1}{c|}{\multirow{2}{*}{$\{TF,T^2F,T^3F\}$}} & \multirow{2}{*}{$\{K,K,Z\}$}            & \multicolumn{1}{c|}{\multirow{2}{*}{Fullness of a level}}                         & \multicolumn{1}{c|}{\multirow{2}{*}{\begin{tabular}[c]{@{}c@{}}Adjacent or\\ same level(s)\end{tabular}}} & \multicolumn{1}{c|}{\multirow{2}{*}{\#files in L0 \textgreater{} $K$}}            & Lazy                                                              \\ \cmidrule(l){9-9} 
                                               & \multicolumn{1}{c|}{}                            & \multicolumn{1}{c|}{}                             & \multicolumn{1}{c|}{}                                    &                                         & \multicolumn{1}{c|}{}                                                             & \multicolumn{1}{c|}{}                                                                                     & \multicolumn{1}{c|}{}                                                             & Greedy                                                            \\ \midrule
Ruskey                                         & \multicolumn{1}{c|}{FLSM}                        & \multicolumn{1}{c|}{$\{T,T,T\}$}                  & \multicolumn{1}{c|}{$\{TF,T^2F,T^3F\}$}                  & $\{K_1,K_2,K_3\}$                       & \multicolumn{1}{c|}{Fullness of a level}                                          & \multicolumn{1}{c|}{\begin{tabular}[c]{@{}c@{}}Adjacent or\\ same level(s)\end{tabular}}                  & \multicolumn{1}{c|}{\#files in L0 \textgreater{} $K_1$}                           & Moderate                                                          \\ \midrule
Moose                                          & \multicolumn{1}{c|}{Generalized LSM}             & \multicolumn{1}{l|}{$\{r_1,r_2,r_3\}$}            & \multicolumn{1}{l|}{$\{r_1F,r_1r_2F,r_1r_2,r_3F\}$}      & $\{\sqrt{r_1},\sqrt{r_2}, \sqrt{r_3}\}$ & \multicolumn{1}{c|}{Fullness of a level}                                          & \multicolumn{1}{c|}{Adjacent levels}                                                                      & \multicolumn{1}{c|}{\#files in L0 \textgreater{} $\sqrt{r_1}$}                    & Not applicable                                                    \\ \midrule
ArceKV                                         & \multicolumn{1}{c|}{ElasticLSM}                  & \multicolumn{1}{c|}{Removed}                      & \multicolumn{1}{c|}{Removed}                             & Removed                                 & \multicolumn{1}{c|}{\begin{tabular}[c]{@{}c@{}}Workload\\ Dependent\end{tabular}} & \multicolumn{1}{c|}{\begin{tabular}[c]{@{}c@{}}Runs from\\ multiple levels\end{tabular}}                  & \multicolumn{1}{c|}{\begin{tabular}[c]{@{}c@{}}Workload\\ Dependent\end{tabular}} & {\begin{tabular}[c]{@{}c@{}}Consistently\\ Optimizing\end{tabular}} \\ \bottomrule
\end{tabular}
}
\label{tab:intro}
\vspace{-4mm}
\end{table*}

Key-value (KV) stores map unique keys to values for fast data access and are widely used in distributed caching, large-scale databases, and cloud services~\cite{elasticache,chang2008bigtable,tidb,cockroachdb,corbett2013spanner,memorystore,workers_kv,sivasubramanian2012amazon}.
Log-Structured Merge Trees (LSM-trees) are fundamental data structures underpinning KV stores, widely supporting modern databases and applications~\cite{cockroachdb,tidb,lakshman2010cassandra,yang2022oceanbase,dgraph,corbett2013spanner}. 
For example, Netflix deploys and optimizes Apache Cassandra~\cite{lakshman2010cassandra}, which is supported by LSM-trees, to effectively handle write-intensive workloads~\cite{netflix}. 
The LSM-tree improves write performance by organizing data as KV entries and deferring expensive in-place updates. It organizes data into multiple hierarchical levels, each with exponentially increasing capacities, structured as sorted runs. New KV entries are first appended to a main-memory buffer (or MemTable); when this buffer fills up, the entries are sorted, compacted, and merged as a larger sorted run into the next level. This background compaction process cascades downwards whenever a level reaches its capacity threshold.

\vspace{1mm}
\noindent{\bf Practical Challenge: Self-adaptation for dynamic workloads.} 
In LSM-tree-based key-value stores, a major challenge lies in online handling dynamically changing workloads. Prior studies~\cite{curino2011workload,gmach2007workload,cao2020characterizing} have shown that real-world applications often exhibit significant workload variability, driven by daily usage patterns and operational shifts. For example, Meta analyzed access patterns from five distinct applications and found that each exhibits highly diverse workload behaviors, with substantial variation occurring even within a single day~\cite{facebookworkload}. This underscores the need to efficiently manage fluctuating ratios of key lookups and entry updates. While many workload-aware methods have been proposed to optimize LSM-tree systems for a given workload, a key challenge remains unresolved for evolving workloads.

Existing workload-aware methods compute a structural configuration, including level capacities, the number of sorted runs, and their sizes to guide compactions and manage write stalls for a given workload. However, when the workload changes, the optimal configuration often changes as well, requiring the system to adapt accordingly. While methods like Moose~\cite{moose} and Wacky~\cite{lsmbush} deliver excellent performance under static workloads, they do not provide mechanisms for transitioning between configurations, making them unsuitable for dynamic workloads.
Naively or greedily resizing runs and merging data during such transitions may introduce latency spikes, as more aggressive write stalls~\cite{dostoevsky2018, ruskey} are often required to reach the desired structure. Dostoevsky~\cite{dostoevsky2018} not only computes a desirable configuration but also introduces a {\it lazy} adaptation strategy, adjusting the size and number of runs in a level only when it is fully compacted into the next. While this approach avoids costly data reorganization, it responds slowly to workload changes and depends on a sufficient number of updates to complete the transition.
In contrast, Ruskey~\cite{ruskey} proposes a {\it middle-ground} strategy called FLSM, which balances between greedy and lazy adaptation. It recalculates the structural configuration when performance degradation is observed and adjusts the active sorted runs during compactions at this level. Although this design accelerates responsiveness, it still relies on sufficient updates to trigger compactions, limiting its ability to adapt promptly under read-intensive workloads. In summary, the existing {\it recomputing and transitioning structure} approaches fail to achieve an excellent tradeoff between responsiveness to the changes and the transitioning overhead.

\vspace{1mm}
\noindent{{\bf Our Vision: Focus on the transition procedure, not on the final structure.}}
Existing approaches are limited by their {\bf rigid transition actions}, often aiming to directly reach a target LSM-tree structure without considering performance during the transition. We argue that under dynamic workloads, the focus should shift from morphing into a pre-defined structure to {\bf consistently optimizing performance} throughout the transition. While it is possible to compute the optimal LSM-tree for a given workload, blindly transitioning toward it may overlook more effective actions that yield better overall system performance.

Building on this insight, we propose two novel designs tailored to dynamic workloads:

\vspace{1mm}
\noindent{\bf {\ourlsm}: Expanding the Transition Action Space.} Existing LSM-trees rely on predefined structural configurations that fix the capacity and number of sorted runs per level, triggering compactions only when level capacity thresholds are exceeded. While this yields predictable costs, it limits flexibility under dynamic workloads. For example, proactively compacting runs across multiple levels— even when they are not full—during a read-intensive phase can further reduce runs and improve read performance.
To enable such flexibility, we introduce {\ourlsm}, which removes rigid limits on level capacities, run counts, and run sizes (Table~\ref{tab:intro}). {\ourlsm} follows a more flexible and workload-dependent policy, treating the LSM-tree as a flexible collection of sorted runs, each tagged with a timestamp, size, and key range. Compactions and write stalls can be triggered or deferred according to the current workload, and may involve runs from one or multiple levels, subject only to preserving the LSM-tree’s intrinsic timestamp ordering. This expanded design allows {\ourkv} to explore a broader set of valid actions, opening more opportunities to optimize performance.

\vspace{1mm}
\noindent{\bf {\oureng}: Lightweight Compaction Evaluation.}
While expanding the action space increases flexibility, it also complicates decision-making. Unlike structurally fixed LSM-trees, where compactions and stalls follow fixed rules with predictable amortized costs, the system must make online decisions in which each action impacts future {\ourlsm} states and costs. This turns the search for a globally optimal action sequence into an intractable, NP-hard problem (see Section~\S\ref{sec:decide}).
To address this, we introduce the \underline{\bf A}daptive \underline{\bf R}untime \underline{\bf C}ompaction \underline{\bf E}ngine ({\oureng}), a score-based evaluation framework that balances both short-term penalties and long-term benefits of compaction actions. With properly tuned parameters, this method restricts the attention to a small set of compactions that must be at least partially involved in the optimal sequence.

Based on {\oureng}, we implement {\ourkv} on top of RocksDB, a widely used industrial LSM-tree storage engine, and evaluate its performance against state-of-the-art compaction policies, including Leveling~\cite{leveldb}, Tiering~\cite{lakshman2010cassandra}, LazyLeveling~\cite{dostoevsky2018}, Ruskey~\cite{ruskey}, and Moose~\cite{moose}. Results show that {\ourkv} achieves high update performance comparable to update-optimized designs while also maintains top-tier read performance compared to read-optimized designs under static workloads. It also adapts rapidly to workload shifts, within 20 million operations and without exhibiting significant latency spikes. Overall, {\ourkv} outperforms RocksDB, the most adaptive among the baselines.
{
Across the two evaluated workloads, {\ourkv} delivers an average performance improvement of 2.17× to 2.92× compared to Tiering and LazyLeveling, and 2.00× and 1.41× relative to 1‑Leveling, which serves as the strongest baseline in our experiments.
}
We further compare {\ourkv} with several industrial-grade databases, including Pebble~\cite{cockroachdb}, RocksDB~\cite{rocksdb}, Cassandra~\cite{lakshman2010cassandra}, and WiredTiger~\cite{wiredtiger}. {\ourkv} delivers over 10$\times$ speedup compared to Cassandra and WiredTiger, and performs 3$\times$ better than Pebble.

\vspace{1mm}
\noindent{\bf Contributions.} In summary, we make the following contributions:
\begin{itemize}[leftmargin=*]
    \item We identify the limitations of existing compaction policies under dynamic workloads and propose a new compaction engine {\oureng} that dynamically selects the most effective compaction and write stall threshold to adaptively balance read and write performance.
    \item We design a score-based model that efficiently estimates the benefit of each compaction and stall threshold pair, providing a near-optimal solution to the underlying NP-hard decision problem.
    \item We implement {\ourkv} on top of RocksDB and demonstrate its effectiveness through extensive evaluations against several state-of-the-art compaction strategies and industrial databases.
\end{itemize}

\vspace{-3mm}
\section{Background}
\label{sec:background}
\begin{figure*}[t]
    \centering
    \vspace{-8mm}
    \includegraphics[width=\linewidth]{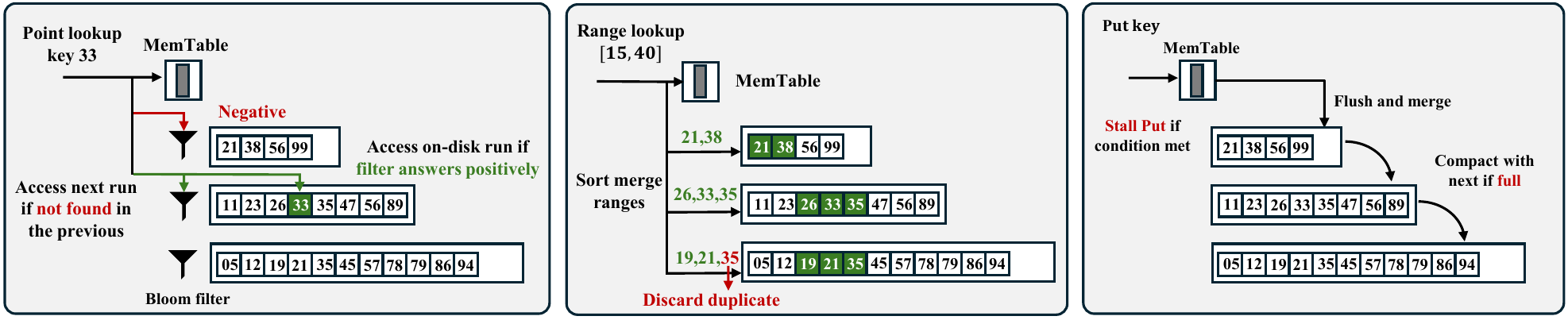}
    \vspace{-8mm}
    \caption{How the three basic operations, point lookup, range lookup, and update, are performed in an LSM-tree system.}
    \vspace{-5mm}
    \label{fig:opintro}
\end{figure*}


\vspace{-1mm}
\subsection{LSM-tree}
\label{sec:LSM-tree}

LSM-tree is a persistent, multi-level indexing structure for key-value stores, which aims to obtain efficient write performance by transforming expensive in-place update into sequential update. All updates, insertions, and deletions are initially turned into a key-value entry and then sorted in a main memory buffer (or MemTable). It will be flushed into the disk as a new sorted run (or SSTable in RocksDB, SST for short) when it is full. These SSTs are organized into several levels, with each level having a capacity $T$ times larger than the previous one. In LSM-trees using a typical Leveling compaction policy, such as Pebble~\cite{cockroachdb}, SSTs at the same level are non-overlapping and collectively form a single sorted run. In contrast, Tiering-based systems like ScyllaDB~\cite{scylladb} allow each level to maintain up to $T$ key-overlapping sorted runs, reducing compaction size and improving write performance. LSM-based systems usually support three basic operations, shown in Figure~\ref{fig:opintro}:

\noindent{\bf Point Lookup.} Given a key, the key-value store returns its associated value if it exists. The search proceeds by scanning each sorted run sequentially, stopping once the value is found. This process relies on the LSM-tree’s timestamp ordering across levels: the smallest timestamp in the $i$-th level must be no smaller than the largest timestamp in the $(i-1)$-th level. Within the same level, runs may have overlapping timestamps. If multiple versions of a key exist at a given level, the system returns the most recent one based on timestamp comparison. Without this cross-level timestamp order, point lookups would require searching all levels for every query, severely degrading performance.
Each sorted run is equipped with a {\it Bloom filter}, an in-memory structure that quickly determines whether a key may exist in the run. Its accuracy is controlled by the bits-per-key (BPK) parameter, representing the ratio of filter memory to the number of keys. The false positive rate (FPR) follows $FPR = O\left(e^{-BPK\cdot(\ln 2)^2}\right)$. Let $s$ be the total number of sorted runs; the I/O cost of a point lookup is then $O(s \cdot FPR + 1)$.

\noindent{\bf Range Lookup.} Different from point lookup, the LSM-tree retrieves all the entries within a specified key range from all the sorted runs. And then it sort merges the results from each sorted runs and produces a final result. Specifically, as most LSM-tree systems leverage {\it iterator} to iteratively produce the final result, which reads the first data block (usually sized one I/O block) from each sorted run and then fetches the entries one by one from each sorted runs. Suppose the search range contains $l$ entries, each of size $E$ bytes, and the I/O block size is $B$ bytes, the I/O cost is $O(s + \frac{lE}{B})$.


\noindent{\bf Update.}  
In an LSM-tree, new key-value pairs are first inserted into an in-memory buffer called the MemTable. Once the MemTable reaches its threshold size, it is flushed to disk as a new sorted run. Updates to existing keys are handled using the same out-of-place insertion mechanism, appending the new version without modifying prior entries. When the size of a level exceeds its predefined capacity, a {\it compaction} is triggered to merge its sorted runs with those in the next level.

Modern LSM-tree key-value systems execute queries and updates on foreground threads, while use background threads to asynchronously handle the flush and compaction when the MemTable or levels become full.

\subsection{Write Stall Controller}

The write stall controller is a critical component in most LSM-tree-based storage systems, including RocksDB~\cite{rocksdb}, Pebble~\cite{cockroachdb}, Cassandra~\cite{lakshman2010cassandra}, and InfluxDB~\cite{influxdb}. It controls the number of sorted runs at the first level (L0) by deliberately stalling incoming writes when they exceed a configurable threshold to maintain a designated number of sorted runs in the system. When a stall is triggered, the new incoming update will be forced to wait for several microseconds. 
Existing workload-aware methods~\cite{moose,ruskey,dostoevsky2018,lsmbush} stall writes when the number of sorted runs in the first level (L0) exceeds the predefined maximum in the structural configuration.
\vspace{-3mm}
\subsection{Open Challenges}
Existing approaches such as Wacky, Moose, Dostoevsky, and Ruskey can derive effective LSM configurations for static workloads, but struggle to transition between configurations with both low cost and high responsiveness. As illustrated in Figure~\ref{fig:flex_example}, a read-intensive workload (90\% reads) favors reducing the number of sorted runs from 10 to 1. A {\it greedy transition} adapts quickly but incurs high overhead by stalling writes when $K_1 = 1$. In contrast, the {\it lazy strategy} and Ruskey defer L0 adjustments, reducing transition cost but causing prolonged performance degradation before convergence.

This trade-off stems from transition mechanisms that focus only on reshaping the structure, rather than maintaining performance throughout the transition. Although a configuration may be optimal for a given workload, existing methods do not optimize system behavior during the transition itself. We argue that {\bf sustaining optimal performance under dynamic workloads requires continuously adapting actions to the current workload and system state.}
\newcommand{\cmark}{\Checkmark}
\newcommand{\xmark}{\XSolidBrush}

\newcommand\halfcmark{\cmark\kern-1.3ex\raisebox{1.0ex}{\rotatebox[origin=c]{125}{\textbf{---}}}}

\newcommand{\onestar}{\xmark}
\newcommand{\twostar}{\halfcmark}
\newcommand{\threestar}{\cmark}
\newcommand{\fourstar}{\cmark\cmark}

\begin{figure*}
\vspace{-3mm}
\includegraphics[width=\linewidth]{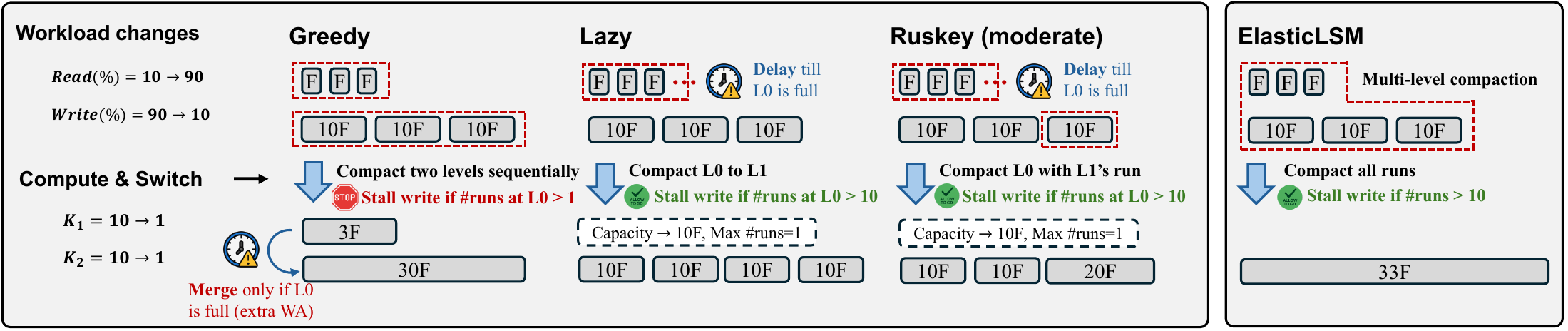}
\scalebox{0.9}{
\begin{tabular}{@{}|r|ccc|ccc|@{}}
\toprule
\multicolumn{1}{|l|}{\multirow{2}{*}{}} & \multicolumn{3}{c|}{Transition Actions}                                                                                                    & \multicolumn{3}{c|}{System Performance}                                                                                                                                                                       \\ \cmidrule(l){2-7} 
\multicolumn{1}{|l|}{}                  & \multicolumn{1}{c|}{\begin{tabular}[c]{@{}c@{}}Write Stall\\ Threshold\end{tabular}} & \multicolumn{1}{c|}{Reduced Runs} & Compacted Bytes & \multicolumn{1}{c|}{\begin{tabular}[c]{@{}c@{}}Transition\\ Overhead\end{tabular}} & \multicolumn{1}{c|}{Responsiveness}            & \begin{tabular}[c]{@{}c@{}}Performance in \\ new workload\end{tabular} \\ \midrule
Greedy                                  & \multicolumn{1}{c|}{\#runs at L0 \textgreater 1}                                     & \multicolumn{1}{c|}{4}            & 33F             & \multicolumn{1}{c|}{\onestar}                                        & \multicolumn{1}{c|}{\threestar} & \threestar                                              \\ \midrule
Lazy                                    & \multicolumn{1}{c|}{\#runs at L0 \textgreater 10}                                    & \multicolumn{1}{c|}{2}            & 10F             & \multicolumn{1}{c|}{\fourstar}                                       & \multicolumn{1}{c|}{\onestar}   & \onestar                                                \\ \midrule
Ruskey                                  & \multicolumn{1}{c|}{\#runs at L0 \textgreater 10}                                    & \multicolumn{1}{c|}{3}            & 20F             & \multicolumn{1}{c|}{\threestar}                                      & \multicolumn{1}{c|}{\twostar}   & \twostar                                                \\ \midrule
{\ourlsm}                              & \multicolumn{1}{c|}{\#runs \textgreater 10}                                          & \multicolumn{1}{c|}{5}            & 33F             & \multicolumn{1}{c|}{\threestar}                                      & \multicolumn{1}{c|}{\fourstar}  & \fourstar                                               \\ \bottomrule
\end{tabular}
}
\vspace{-3mm}
\caption{The example illustrates how existing structural transition policies: Greedy, Lazy, and Moderate (Ruskey), perform during and after the transition, compared with {\ourlsm}’s continuous optimization approach. ``Responsiveness'' denotes the speed at which each method completes the transition. Performance ratings are denoted as follows: \onestar = worst, \twostar = mediocre, \threestar = good, and \fourstar = best.}
\vspace{-6mm}
\label{fig:flex_example}
\end{figure*}


\vspace{-3mm}
\section{{\oureng}: Adaptive Compaction Decision}
\label{sec:arce}

To achieve this, we first decouple compaction behavior from rigid structural configuration parameters, such as fixed level capacities and prescribed run counts, which traditionally enforce a static and inflexible compaction schedule. We then introduce {\ourlsm},
an enhanced LSM-tree design that enables the system to {\bf compact runs selected across levels} and {\bf initiate compactions at workload-dependent timing}, yielding two principal benefits:
{
\begin{itemize}[leftmargin=*]
    \item {\bf Flexible Run Selection:} Merging sorted runs across multiple levels into one run in a single compaction improves responsiveness to read-intensive workloads and helps reduce write amplification. Also, selectively merging runs within a single level during write-intensive workloads reduces compaction overhead while slightly improving read performance.
    \item {\bf Workload Dependent Timing:} By permitting compaction to be scheduled dynamically in response to workload conditions, the system can delay or advance compactions and write stalls as needed. This flexibility improves adaptability to workload shifts and mitigates the risk of performance bottlenecks under write-heavy scenarios.
\end{itemize}
}
For example, as shown in Figure~\ref{fig:flex_example}, by removing structural configuration parameters, {\ourlsm} can compact all runs across levels in a single operation while simultaneously raising the write stall threshold. This combination avoids transition costs and delivers even better responsiveness than the Greedy approach.

In the following, we first describe how to identify action candidates after removing parameters (Section~\S\ref{sec:unlocked_compaction}), then present a theoretical model of system cost under this setting (Section~\S\ref{sec:compaction_est}) to guide {\oureng} in selecting the most suitable actions over time (Section~\S\ref{sec:decide}).


\begin{figure*}
    \centering
    \vspace{-4mm}
    \includegraphics[width=0.97\linewidth]{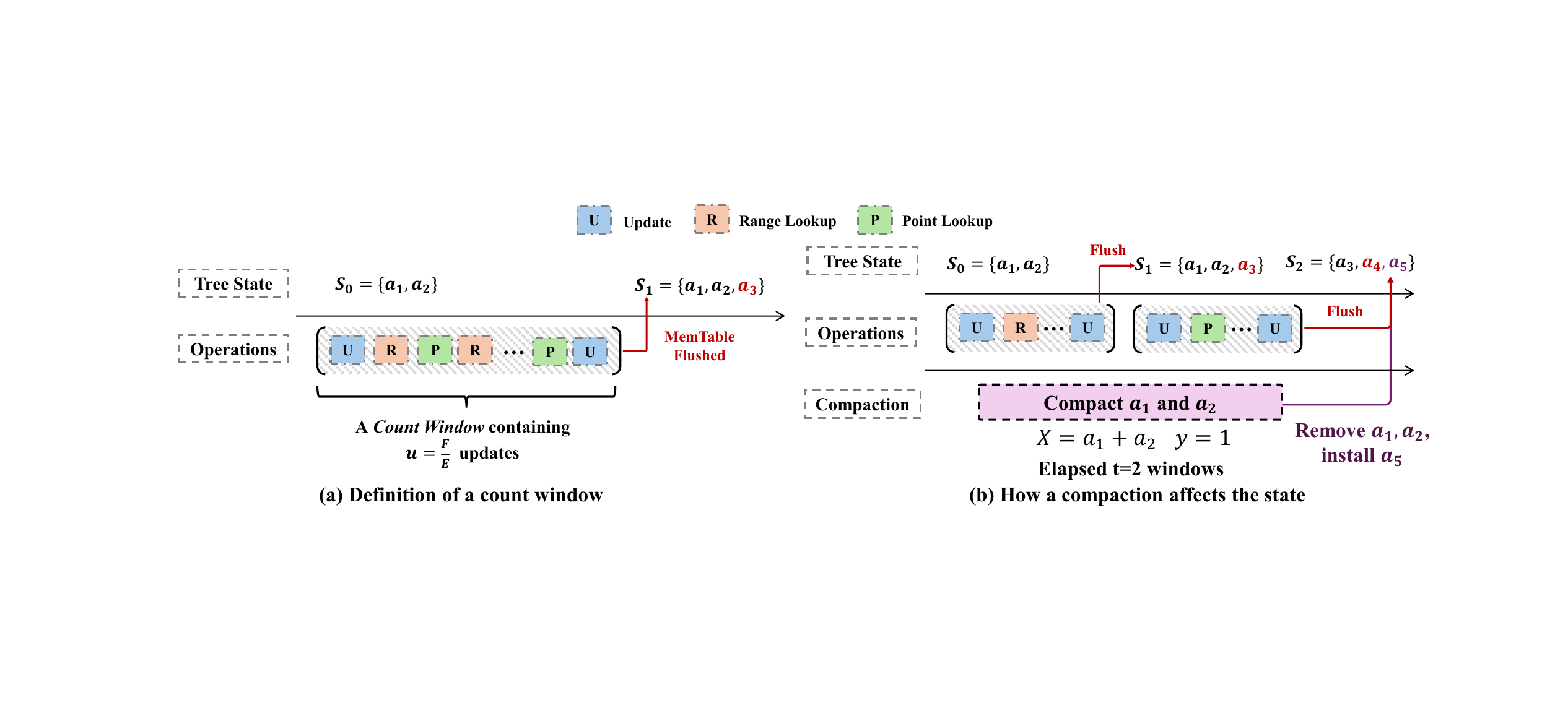}
    \vspace{-4mm}
    \caption{Illustration of how the tree state evolves when the MemTable is flushed or a compaction finishes.}
    \vspace{-4mm}
    \label{fig:cost_est}
\end{figure*}

\subsection{{\ourlsm}: Expanded Action Space}
\label{sec:unlocked_compaction}
{\ourlsm} maintains a collection of sorted runs across levels, each potentially varying in size and count. Without fixed structural parameters on level capacities or maximum run counts, the system must explicitly decide when and how to perform its two core management actions: compaction and write stall. Write stall in {\ourlsm} is straightforward: updates are throttled only when the total number of sorted runs exceeds a tunable threshold $c$, with a stalling rate $k$. This flexibility allows the system to better balance read and write throughput. Both parameters can be tuned independently, as detailed in Section~\S\ref{sec:param_search}. In the following, we elaborate the more complex action -- compaction.

\begin{figure}
\centering
  \includegraphics[width=\linewidth]{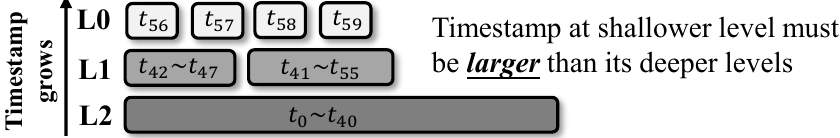}
  \vspace{-6mm}
  \caption{Example of timestamp ranges of different levels.}
  \label{fig:ts_order}
  \vspace{-2mm}
\end{figure}
\noindent{\bf Extensive Compaction Options.}
Any level can contain an arbitrary number of sorted runs of varying sizes after removing structural constraints like level capacities, sorted runs number, and run sizes. However, this flexibility does not imply that we can freely merge any subset of runs. The core requirement of an LSM-tree is to {maintain timestamp
ordering across levels: timestamp at shallower level must be larger than that in deeper levels, while within the same level, sorted runs can have overlapping timestamps (see Section~\S\ref{sec:background}). As shown in Figure~\ref{fig:ts_order}, within the same level, the timestamp ranges of runs can be overlapped while the bottom level must have disjoint and smaller timestamps compared to the top level.
Additionally, we restrict compactions to proceed downward, following the LSM-tree tradition, to avoid complicating the timestamp order of runs within a level.
}

Based on these rules, we identify three compaction patterns that produce valid compaction candidate set:

\begin{itemize}[leftmargin=*]
    \item {\bf Pattern 1 (Intra-level)}: Compact any more than one sorted runs at $i$-th level, and place the result to the $i$-th level.
    \item {\bf Pattern 2 (Adjacent-level)}: Compact all the sorted runs from the $i$-th level with zero or more sorted runs at the $(i+1)$-th level, and place the result to the $(i+1)$-th level.
    \item {\bf Pattern 3 (Multi-level)}: Compact all sorted runs from the $i$-th to the $j$-th level ($j>i+1$) with zero or more sorted runs at the $(j+1)$-th level and place the result at the $(j+1)$-th level.
\end{itemize}

In general, Pattern 1 enables intra-level compaction, Pattern 2 performs traditional adjacent-level compaction, and Pattern 3 supports multi-level compaction. While these patterns enable a wide range of compaction candidates, the resulting candidate set can be extremely large and computationally expensive to process exhaustively. To address this, we apply heuristic pruning. Our observation is that, for similar sizes of compacted data, reducing a greater number of sorted runs generally yields better lookup performance. Therefore, for Pattern 1, instead of enumerating all possible combinations of runs within a level, we first sort the runs by sizes in ascending order. We then iteratively build compaction candidates by starting with the smallest run and incrementally adding one more run at a time, continuing until all runs are included. Each intermediate compaction is added to the candidate set. A similar strategy is applied for Pattern 2, where the runs in $(i+1)$-th level are also sorted and incrementally included. For Pattern 3, although a similar incremental approach can be applied, the resulting compaction candidate set can still grow to an enormous size when the total number of levels is large. Therefore, in practice, we typically limit the number of levels to fewer than 8\footnote{The default number of levels in RocksDB is 7.}. By doing this pruning, {\oureng} is able to rapidly find the valid compaction set in 30us in our experiment.

\subsection{System Cost Modeling}
\label{sec:compaction_est}

Since {\ourlsm} greatly expands the action space, it is crucial to understand how different compaction strategies and write stall parameters influence overall performance before making decisions. In traditional LSM-trees, operational costs are straightforward to predict because compactions and stalls follow fixed patterns. In contrast, our flexible design makes cost estimation more challenging, as the tree state (i.e., the sorted runs in the tree) can evolve by more flexible and unpredictable actions. To address this, we introduce a {\it Windowed-State Cost Modeling} method, which partitions the long running operation sequence into multiple state-stable windows, where tree state is generally unchanged. We then estimate the three operational costs within each window and define the rules for state transitions between consecutive windows.

\noindent{\bf Count Window: Maintaining a Stable Tree State.}
\begin{figure}
\includegraphics[width=\linewidth]{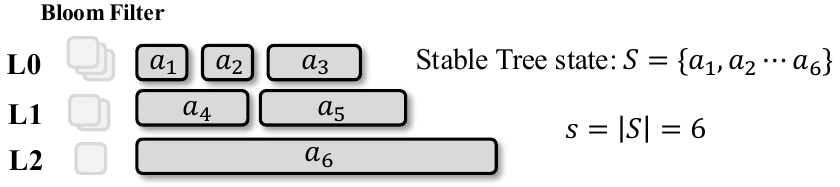} 
\vspace{-6mm}
\caption{An example of an {\ourlsm} within a count window.}
\vspace{-3mm}
\label{fig:stabletree}
\end{figure}
As discussed in Section~\S\ref{sec:background}, both range and point lookup costs depend on the number of sorted runs. In {\ourlsm}, removing structural constraints enables more flexible compactions, but also makes the number of runs highly unpredictable. We observe that the most frequent change in run count occurs when a full MemTable is flushed to the first level (L0), whereas compactions, although they also modify the run count, generally take longer to complete.
Based on this observation, we partition foreground operations into consecutive windows, each containing $u = F/E$ updates, where $F$ is the MemTable size and $E$ is the entry size. We term these {\it count windows} (or simply, windows), inspired by stream processing techniques~\cite{flinkwindows}. The number of range lookup and point lookup within a window are denoted as $r$ and $p$ respectively. And naturally, we can describe the workload pattern by $(r,u,p)$ tuple.
Within a window, we maintain a relatively stable {\it tree state} $S$, defined as the set of sorted runs and their sizes. 
As illustrated in Figure~\ref{fig:stabletree}, the example shows a stable tree state within a window containing six sorted runs of sizes $a_1$ to $a_6$ across three levels.
This stable state allows us to estimate the cost of the three primary operations within the window as follows.

\noindent{\bf Operational Cost in a Window.}
For a point lookup, the LSM-tree may scan up to all $s=|S|$ runs to locate the target key. Each run is equipped with a Bloom filter with false positive rate $\alpha$, so I/O to access a data block is required only when the filter returns a positive result. In the worst case, exactly one run yields a true positive, while the others incur I/O only on false positives with probability $\alpha$. The resulting cost is given in Equation~\ref{eq:point}, where $I_r$ denotes the I/O time to access a data block.

For a range lookup, the system first locates the start position and retrieves the corresponding block from each run, incurring a cost of $s \cdot I_r$. It then sequentially scans $l$ entries (range length) from each run, with I/O cost $lE/B \cdot I_r$, where $E$ is the entry size and $B$ the data block size. Since this scanning phase depends only on $l$ and not on the LSM-tree state, we omit it from subsequent optimization (see Equation~\ref{eq:range}).

For updates, prior methods tie write stalls to compaction, with stall time proportional to compacted bytes as dictated by structural constraints. In contrast, {\ourlsm} decouples compaction from stalling: updates are slowed by a tunable rate $k$ only when $s$ exceeds an independent threshold $c$. The update cost is thus the flush I/O cost plus the stall penalty $k \cdot \mathbb{I}(s>c)$, where $\mathbb{I}$ returns $1$ if $s>c$ and $0$ otherwise, and $I_w$ is the I/O time to write a block (Equation~\ref{eq:update}).

\begin{align}
    \textbf{Point Lookup Cost \quad }P(s) &=(\alpha \cdot s+1)\cdot I_r\label{eq:point}\\
    \textbf{Range Lookup Cost \quad }  R(s) & =s\cdot I_r \label{eq:range}\\
    \textbf{Update Cost \quad }U(s) &={(F/B)}\cdot I_w+k\cdot \mathbb{I}(s>c) \label{eq:update}
\end{align}

\noindent{\bf Evolving Tree State Between Windows.}
Once the cost within a single window is known, estimating the cost of the $i$-th window requires understanding how the run count in the tree state changes from $s_{i-1}$ to $s_i$. Such changes occur through two background actions: MemTable flushes and compactions.
As shown in Figure~\ref{fig:cost_est}(a), flushing a MemTable simply adds a new sorted run of size $a_3$ to the state, yielding $s_{i} = s_{i-1} + 1$. In contrast, compaction alters the tree state more intricately, since its completion time is uncertain and typically not aligned with window boundaries.
To address this, we note that a compaction in the background thread completes when the total I/O time of foreground operations equals (or exceeds) the compaction’s I/O time when having sufficient I/O bandwidth. Specifically, if a compaction of size $X$ bytes starts in the $i$-th window, it will finish in the $(i+t)$-th window, where the cumulative I/O cost of foreground operations over $t$ windows matches the compaction’s I/O time. The foreground I/O time in $t$ windows without other concurrent compactions is given in Equation~\ref{eq:tcost}.
To preserve a stable tree state within each window, we consider the compaction to take effect in the next window after completion. The value of $t$ is computed using Equation~\ref{eq:t}. Experimental results (Figures~\ref{fig:internal_param}(c) and (d)) show that this rounding has minimal impact on the accuracy of theoretical cost model.
\begin{align}
    \label{eq:tcost}
    f(s,t)=\sum_{i=0}^{t-1} r\cdot R(s+i)+u\cdot U(s+i)+ p\cdot P(s+i)
\end{align}

\begin{align}
    \label{eq:t}
    t = &\min \Bigl\{ t \in \mathbb{Z}^+ \,\Big|\, f(s,t)\geq \frac{X}{B}(I_r+I_w) \Bigr\}
\end{align}
\begin{example}
As shown in Figure~\ref{fig:cost_est}(b), for a compaction of size $X$ that removes $y$ sorted runs, if the estimated completion time is after 2 windows, its effect will be applied in the third window by removing the compacted runs (e.g., $a_1$ and $a_2$) and installing the result (e.g., $a_5$).
\end{example}
The number of sorted runs of windows evolves by:

\begin{align}
    s_{i+1}=\begin{cases}
        s_i+1,&\text{No compaction completes at i+1 window}\\
        s_i+1-y,&\text{Compaction reducing y runs completes}
    \end{cases}
\end{align}
\subsection{{\oureng}: Decide the Intermediate Compaction}
\label{sec:decide}

\noindent{\bf Objective Function.}
Based on the cost model and state-evolution rules defined above, we can express the average cost for a given workload $(r,u,p)$ with stall parameters $c$ and $k$, after performing $m$ compactions, as:

\begin{align}
\label{eq:cost}
C = \frac{\sum_{i=1}^{m} f(s_i, t_i)}{\sum_{i=1}^{m} t_i (r + u + p)}
\end{align}

subject to the update rule:
\begin{align}
s_{i+1} = s_i + t_i - y_i
\end{align}
Here, each compaction has $X_i$ bytes, reduces $y_i$ sorted runs, spans over $t_i$ count windows, and starts with $s_i$ sorted runs. The function $f(s_i, t_i)$ represents the cumulative cost over window $t_i$, as defined earlier.
This problem is fundamentally a search problem to find out $m$ compaction to minimize the average cost. Evaluating only short-term compaction candidates (i.e., small $m$) is computationally efficient, but tends to favor smaller, quickly completed compactions and larger stall thresholds $c$, which yield short-term benefits by reducing run count more rapidly. However, such strategies may overlook larger compactions that, although expensive upfront, offer substantial long-term benefits. For example, merging two 40GiB runs may yield sustained lookup improvements for the next 40GiB of inserted data.
Exploring deeper compaction sequences to capture these long-term gains introduces significant computational overhead and can be proven to be NP-hard. Formal proof is provided in Section~\ref{sec:proof-details}.
\begin{lemma}
\label{lmm:nphard}
Deciding $m$ compactions to minimize Equation~\ref{eq:cost} is NP-hard.
\end{lemma}
\begin{figure}
    \centering
    \includegraphics[width=\linewidth]{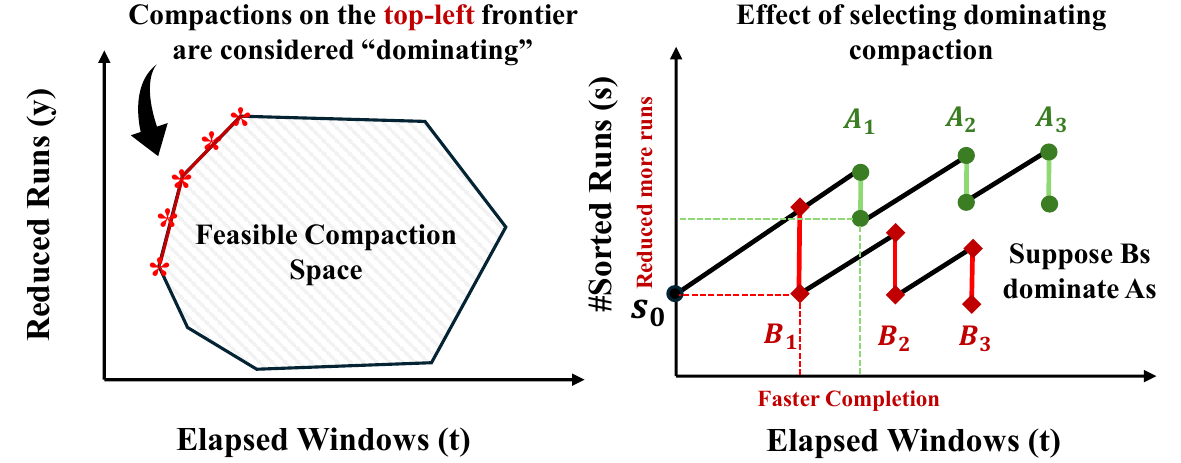}
    \vspace{-6mm}
    \caption{The left panel defines dominating compactions, while the right panel illustrates their benefits.}
    \vspace{-3mm}
    \label{fig:pareto}
\end{figure}

Fortunately, it is unnecessary to determine the full sequence of $m$ compactions in advance.
Instead, we only need to identify the first compaction to execute at each decision point. This raises a key question: Can we design a principled method to quantify both the short-term penalty and long-term benefit of an intermediate compaction candidate?

\noindent{\bf Short-Term Effect.}
The immediate drawback—or penalty—of executing a compaction is that it occupies a background compaction worker, potentially causing SSTs to accumulate at L0. This accumulation can degrade read performance and even trigger a write stall. We model the short-term cost as:
\begin{align}
\label{eq:harm}
E_s(s, t) &= \overbrace{I_r \cdot t \cdot (r + \alpha \cdot p)}^{\text{Read slowdown}} + \underbrace{uk \cdot \max\big(0, s + t - c\big)}_{\text{Write stall penalty}}
\end{align}
Here, $t$ denotes the estimated duration of compaction with sorted runs $s$ in the system.

\noindent{\bf Long-Term Effect.}
Compaction reduces the number of sorted runs, which benefits all future reads within the current decision window. We define the long-term benefit of a compaction that reduces $y$ sorted runs as:
\begin{align}
E_l(y) = (r + \alpha \cdot p) \cdot I_r \cdot y
\end{align}

\noindent{\bf Effectiveness Score.}
By integrating both effects, we define the overall effectiveness of a compaction spanning $t$ windows and reducing $y$ runs as:
\begin{align}\label{eq:effect}
E(s,t,y) = M \cdot E_l(y) - E_s(s,t)
\end{align}
The parameter $M$ scales the long-term benefit and is determined by the current tree state, workload characteristics, and write stall threshold. Section~\ref{sec:param_search} provides guidance on selecting the appropriate $(M,c,k)$ under different scenarios. Given a fixed $(M,c,k)$, {\oureng} can select the compaction with the highest effective score among many compaction candidates.

\noindent{\bf Optimality Analysis.}
The effectiveness score not only significantly improves the efficiency of compaction selection but also reveals an important structural property among compaction candidates—domination. Formally, we say that compaction $A$ dominates compaction $B$ (denoted as $B \prec A$) if and only if $A$ reduces more sorted runs as $B$ while requiring less compaction time.
Using the score-based evaluation defined in Equation~\ref{eq:effect}, {\oureng} ensures that only non-dominated candidates are selected under any given parameter configuration $(M, c, k)$. We refer to these as {\it dominating compactions}, which collectively form the left frontier in a two-dimensional space, where the x-axis represents elapsed time $t$ and the y-axis represents the number of reduced sorted runs $y$, as illustrated in Figure~\ref{fig:pareto}.

\begin{lemma}
    \label{lemma:domination}
    If $A\prec B$, the effectiveness score of $A$ is less than $B$.
\end{lemma}
\begin{proof}
    The long-term effect of $B$ will increase by $(r+\alpha\cdot p)\cdot I_r\cdot (y_2-y_1)$, while the short-term effect of $B$ will decrease by at least $I_r\cdot (t_1-t_2)\cdot (r+\alpha\cdot p)+uk\cdot \max(0,t_1-t_2+s-c)$. Therefore, the effectiveness score of $B$ is larger than $A$.
\end{proof}
In Section~\ref{sec:proof-details}, we show that every dominating compaction can be selected by some $(M, c, k)$, and at least part of the dominating compaction must occur in the optimal sequence.

\subsection{Parameter Selection}
\label{sec:param_search}
\begin{algorithm}[t]
    \caption{FindBestParams($M,c,k$)}
    \label{alg:simulate}
    \KwIn{Current tree state $S$ and workload $(r, u, p)$}
    \KwOut{Best parameters $(M,c,k)$}

    bestCost $\leftarrow$ $\infty$ \;
    bestM, bestc, bestk $\leftarrow$ null \;

    \ForEach{valid $(M,c,k)$}{
        totalCost $\leftarrow$ 0 \;
        $S' \leftarrow S$ \;
        
        \For{$i \leftarrow 0$ \KwTo MaxIterTime}{
            Select compaction reducing $y$ runs and spanning $t$ windows based on $(M, c, k)$ \;
            totalCost $\leftarrow$ totalCost + f(|S'|,t)\;
            totalOps $\leftarrow$ totalOps + $t\cdot (r+u+p)$\;
            $S'$ $\leftarrow$ $S'$ removes compacted runs and installs result \;
        }
        
        avgCost $\leftarrow$ totalCost / totalOps
        
        \If{avgCost $<$ bestCost}{
            bestM $\leftarrow M$ \;
            bestc $\leftarrow c$ \;
            bestk $\leftarrow k$ \;
            bestCost $\leftarrow$ avgCost \;
        }
    }

    \Return{$(bestM, bestc,bestk)$}
\end{algorithm}

Achieving approximately optimal compaction selection requires properly setting the parameters. However, determining the optimal values of the three parameters $(M, c, k)$ over time is itself an NP-hard problem. Fortunately, it is not necessary to determine all parameters simultaneously. Instead, we only need to ensure that the parameter values chosen at each decision point are suitable, and we can update them periodically as the system evolves.
To this end, we adopt a simple yet effective simulation-based approach. We iteratively explore a wide range of candidate $(M, c, k)$ combinations and evaluate their effectiveness by simulating continuous compaction decisions. For each configuration, we estimate the average system cost over a sufficiently long period. The parameter set yielding the lowest cost is then selected. This process is detailed in Algorithm~\ref{alg:simulate}.

The underlying intuition is that when the tree state (e.g., total data volume and number of sorted runs) and the workload remain relatively stable, there exists a tuple of parameters $(M, c, k)$ that can continuously guide the selection of the most suitable compactions to minimize system cost. A new parameter tuple is required only when any of them varies beyond a predefined recomputing threshold $d$ ($d \in (0,1)$). In our implementation, we use $d=0.1$, which strikes a balance between simulation overhead and responsiveness, ensuring satisfactory performance without frequent re-selection. A detailed evaluation of this threshold is provided in Section~\S\ref{sec:evaluate}.
To further reduce simulation time, we employ several optimization techniques, including candidate pruning and multi-threaded computation, as described in Section~\S\ref{sec:arcekv}.
\vspace{-3mm}
\section{{\ourkv}: Workload-driven KV Store}
\label{sec:arcekv}
\begin{figure}[t]
    \centering
    \includegraphics[width=\linewidth]{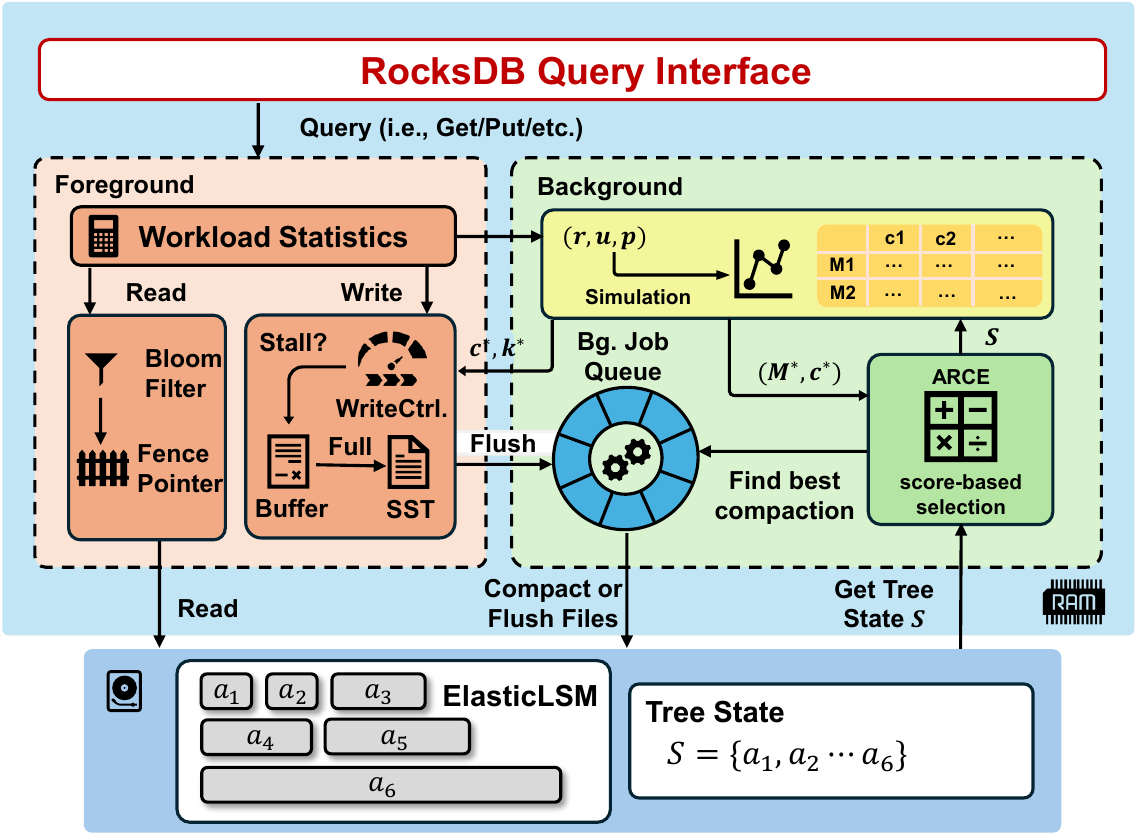}
    \vspace{-8mm}
    \caption{Overview of the architecture of {\ourkv}}
    \vspace{-3mm}
    \label{fig:arch}
\end{figure}

As shown in Figure~\ref{fig:arch}, {\ourkv}, built upon {\oureng} introduced earlier, incorporates four key components: Workload Statistics, {\ourlsm}, Background Simulation, and {\oureng}’s score-based compaction selection. The Workload Statistics module monitors operations and reports windowed counts $(r, u, p)$ every one million operations. The Background Simulation module implements Algorithm~\ref{alg:simulate}, outputting $c$ and $k$ to RocksDB’s write controller to manage write stalls, and forwarding $M$, $c$, and $k$ to {\oureng}’s score-based selection module to determine optimal compactions according to Equation~\ref{eq:effect}. {\ourlsm} manages SSTables organized into sorted runs and levels, which {\oureng} dynamically compacts based on these decisions. All components are implemented within RocksDB. The following section details the implementation.


\noindent{\bf Parallel Simulation.} Unlike tree traversal, the score-based simulation is inherently parallelizable, as each $(M, c,k)$ tuple can be evaluated independently. Distributing the computation across multiple threads can significantly accelerate the algorithm. By default, {\ourkv} uses 16 threads to run these simulations, which complete well within a single window. Additionally, the effectiveness scores are computed through linear transformations of compaction size $X$, reduced runs $y$, and elapsed windows $t$. This structure allows us to vectorize the score computation across all compaction candidates using the \texttt{Eigen} library~\cite{eigen}. \texttt{Eigen} applies SIMD (Single Instruction, Multiple Data) optimizations under the hood, further improving simulation efficiency.

\noindent{\bf Parameter Pruning.}
To reduce computational overhead, we adopt coarse-grained parameter tuning. For $M$, we use a step size of 5, as nearby values yield similar compaction choices; for $c$, we use a step size of 2, since closely spaced $(M, c)$ pairs produce comparable results. We also set upper bounds: $M$ is capped at the smallest value selecting the compaction with the largest reduction in sorted runs (upper-right candidate in Figure~\ref{fig:pareto}(a)), and $c$ is limited to less than $4\times$ the current run count, as exceeding this should already trigger re-selection based on the change threshold in the previous section.
For the write stall penalty $k$, performance changes significantly only when it is doubled or halved. Thus, we initialize $k = 6$ (RocksDB’s default stall rate) and test two additional values by successive doubling, as finer granularity provides diminishing returns.
Finally, we cap the simulation iterations for each parameter tuple at \texttt{MaxIterTime} (typically 400) to ensure completion before the tree state drifts, while keeping the duration long enough to capture long-term benefits. Under a balanced workload ($r = u = p$), a window lasts over 200 ms and simulation completes within 150ms, keeping the system responsive during parameter selection.

\noindent{\bf Multi-threading Extension.}
The cost model in Section~\S\ref{sec:compaction_est} assumes a single foreground thread and one background compaction worker, whereas real systems typically execute multiple foreground threads concurrently. Although the per-operation I/O cost remains unchanged, parallel execution reduces the number of elapsed compaction windows. Following Cosine~\cite{chatterjee2021cosine}, we model this effect using Amdahl’s Law~\cite{Amdahllaw}. With $\eta$ cores and parallelizable fraction $\phi$, the speedup is $g=\frac{1}{1-\phi(1-1/\eta)}$, yielding an adjusted compaction duration $t' = t/g$. Based on profiling, we set $\phi=0.5$.

To support multiple background workers, {\ourkv} tracks available compaction threads. When multiple workers are idle, newly flushed SSTables are assumed to compact immediately and no penalty is applied; otherwise, a delay penalty is introduced. The system also monitors thread, memory, and I/O utilization, applying a large penalty when resources become saturated.

\vspace{-1mm}
\section{Evaluation}
\label{sec:evaluate}
\newcommand{\workloadfirst}{Workload~I}
\newcommand{\workloadsecond}{Workload~II}
\newcommand{\workloadthird}{Workload~III}

\begin{table}[t]
    \centering
    \caption{Operation Ratios Composition}
    \vspace{-4mm}
    \begin{tabular}{rllllllllll}
    \hline
    \, & A & B & C & D & E & F & G & H & I & J\\ \hline
    range(\%) & 98 & 1 & 1 & 49 & 2 & 49 & 40 & 40 & 20 & 33\\
    update(\%)     & 1 & 98 & 1 & 2 & 49 & 49 & 40 & 20 & 40 & 33\\
    point(\%) & 1 & 1 & 98 & 49 & 49 & 2 & 20 & 40 & 40 & 33\\
    \hline
    \end{tabular}
    \vspace{-3mm}
    \label{tab:workload}
\end{table}

This section presents the experimental evaluation of {\ourkv}. All experiments are conducted on a machine equipped with an Intel Core i9-13900K CPU (5.40GHz), 128GB of RAM, and a 1TB NVMe SSD, running 64-bit Ubuntu 22.04 with an ext4 file system. To simulate realistic deployment scenarios, where not all system memory is allocated to RocksDB (e.g., TiKV recommends allocating 70\%~\cite{tikv}), We follow Disco~\cite{disco} in using the cgroup command to cap total memory usage at 75GiB. In terms of baselines, we include both widely adopted compaction strategies, such as Leveling~\cite{leveldb}, 1-Leveling~\cite{rocksdb}, and Tiering~\cite{lakshman2010cassandra}, and recently published academic state-of-the-art LSM-tree systems with available artifacts, including LazyLeveling~\cite{dostoevsky2018}, Moose~\cite{moose}, Ruskey~\cite{ruskey}, and CAMAL~\cite{camal}. These systems are all capable of handling point lookups, range queries, and updates issued in random order and speed using a unified handler thread.
To ensure fair comparison and practical generalizability, we adopt RocksDB’s default configuration, using one background compaction thread and one flush thread. Additionally, to broaden the evaluation of effectiveness, we include industrial-grade systems such as Pebble~\cite{cockroachdb}, WiredTiger~\cite{wiredtiger}, and Cassandra~\cite{lakshman2010cassandra} as baseline

\begin{figure*}
    \centering
    \vspace{-2mm}
    \includegraphics[width=0.97\textwidth]{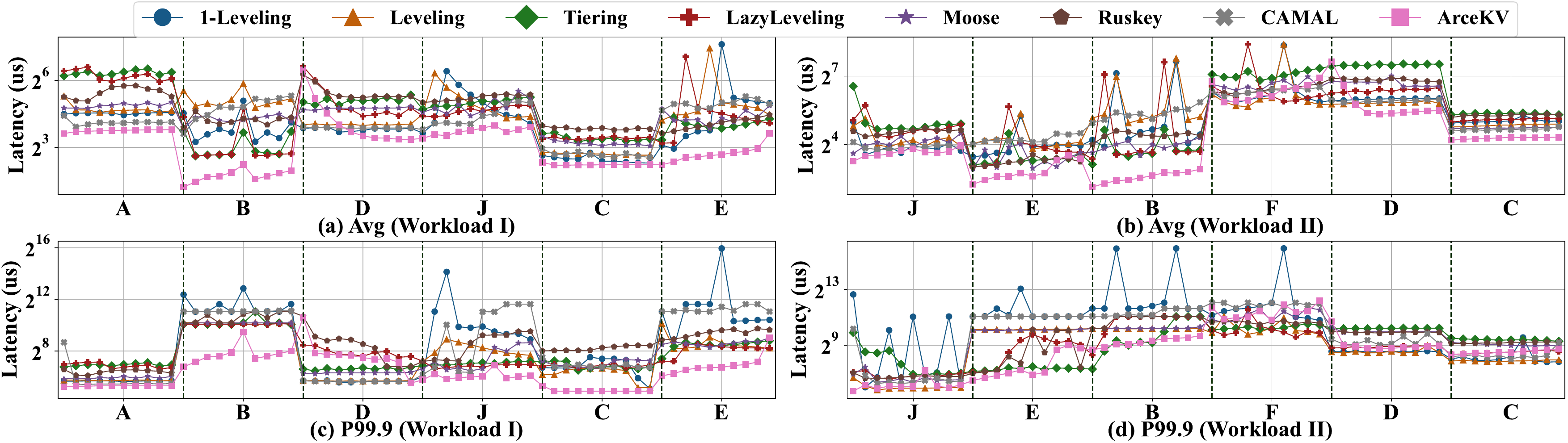}
    \vspace{-3mm}
    \caption{Average latencies and P99.9 latencies for all the methods under {\workloadfirst} and {\workloadsecond}.}
    \vspace{-2mm}
    \label{fig:prim_workload_test}
\end{figure*}
\begin{figure*}
\vspace{-2mm}
\begin{minipage}[b]{.4\linewidth}
\hspace{-3mm}
\raisebox{40mm}{
    \centering
    \small
    \scalebox{0.75}{
    \renewcommand{\arraystretch}{0.95}
    \begin{tabular}{|l|ccccccc|}
\hline
                       & \multicolumn{7}{c|}{\textbf{Workload I}}                                                                                                                                               \\ \cline{2-8} 
\multicolumn{1}{|c|}{} & A                    & B                     & D                     & J                    & C                    & \multicolumn{1}{c|}{E}                    & AvgTput              \\ \hline
1-L                    & 3.52x                & 2.29x                 & {\ul \textbf{3.20x}}  & 1.14x                & 2.60x                & \multicolumn{1}{c|}{1.04x}                & 1.45x                \\
Lvl                    & 3.07x                & 1.00x                 & 2.78x                 & 1.11x                & 2.21x                & \multicolumn{1}{c|}{1.00x}                & 1.19x                \\
Tier                   & 1.00x                & 4.21x                 & 1.31x                 & 1.16x                & 1.35x                & \multicolumn{1}{c|}{2.50x}                & 1.03x                \\
LL                     & 1.08x                & 3.65x                 & 1.20x                 & 1.47x                & 1.35x                & \multicolumn{1}{c|}{1.33x}                & 1.00x                \\
MSE                    & 2.76x                & 1.85x                 & 1.68x                 & 1.26x                & 1.57x                & \multicolumn{1}{c|}{1.91x}                & 1.40x                \\
RKY                    & 1.82x                & 1.92x                 & 1.00x                 & 1.00x                & 1.00x                & \multicolumn{1}{c|}{2.39x}                & 1.07x                \\
CAM                    & 4.65x                & 1.25x                 & 3.07x                 & 1.68x                & 2.23x                & \multicolumn{1}{c|}{1.29x}                & 1.57x                \\
\textbf{OURS}          & {\ul \textbf{6.04x}} & {\ul \textbf{10.60x}} & 1.96x                 & {\ul \textbf{2.81x}} & {\ul \textbf{3.08x}} & \multicolumn{1}{c|}{{\ul \textbf{5.89x}}} & {\ul \textbf{2.92x}} \\ \hline
                       & \multicolumn{7}{c|}{\textbf{Workload II}}                                                                                                                                              \\ \cline{2-8} 
\multicolumn{1}{|c|}{} & J                    & E                     & B                     & F                    & J                    & \multicolumn{1}{c|}{C}                    & AvgTput              \\ \hline
1-L                    & 2.18x                & 1.18x                 & 1.26x                 & 1.45x                & 2.84x                & \multicolumn{1}{c|}{1.24x}                & 1.53x                \\
Lvl                    & 1.74x                & 1.00x                 & 1.00x                 & 1.53x                & {\ul \textbf{3.19x}} & \multicolumn{1}{c|}{1.42x}                & 1.51x                \\
Tier                   & 1.00x                & 1.72x                 & 4.34x                 & 1.00x                & 1.00x                & \multicolumn{1}{c|}{1.00x}                & 1.00x                \\
LL                     & 1.23x                & 1.18x                 & 1.48x                 & 1.37x                & 2.17x                & \multicolumn{1}{c|}{1.20x}                & 1.37x                \\
MSE                    & 1.99x                & 1.76x                 & 4.10x                 & 1.39x                & 1.73x                & \multicolumn{1}{c|}{1.54x}                & 1.53x                \\
RKY                    & 1.45x                & 2.00x                 & 3.00x                 & 1.59x                & 1.58x                & \multicolumn{1}{c|}{1.05x}                & 1.42x                \\
CAM                    & 2.30x                & 1.01x                 & 1.54x                 & 1.67x                & 2.99x                & \multicolumn{1}{c|}{1.61x}                & 1.72x                \\
\textbf{OURS}          & {\ul \textbf{2.74x}} & {\ul \textbf{2.60x}}  & {\ul \textbf{10.62x}} & {\ul \textbf{1.64x}} & 2.79x                & \multicolumn{1}{c|}{{\ul \textbf{2.10x}}} & {\ul \textbf{2.17x}} \\ \hline
\end{tabular}}}
\end{minipage}%
\begin{minipage}[b]{.5\linewidth}
\hspace{1mm}
\vspace{15mm}
\includegraphics[width=1.\linewidth]{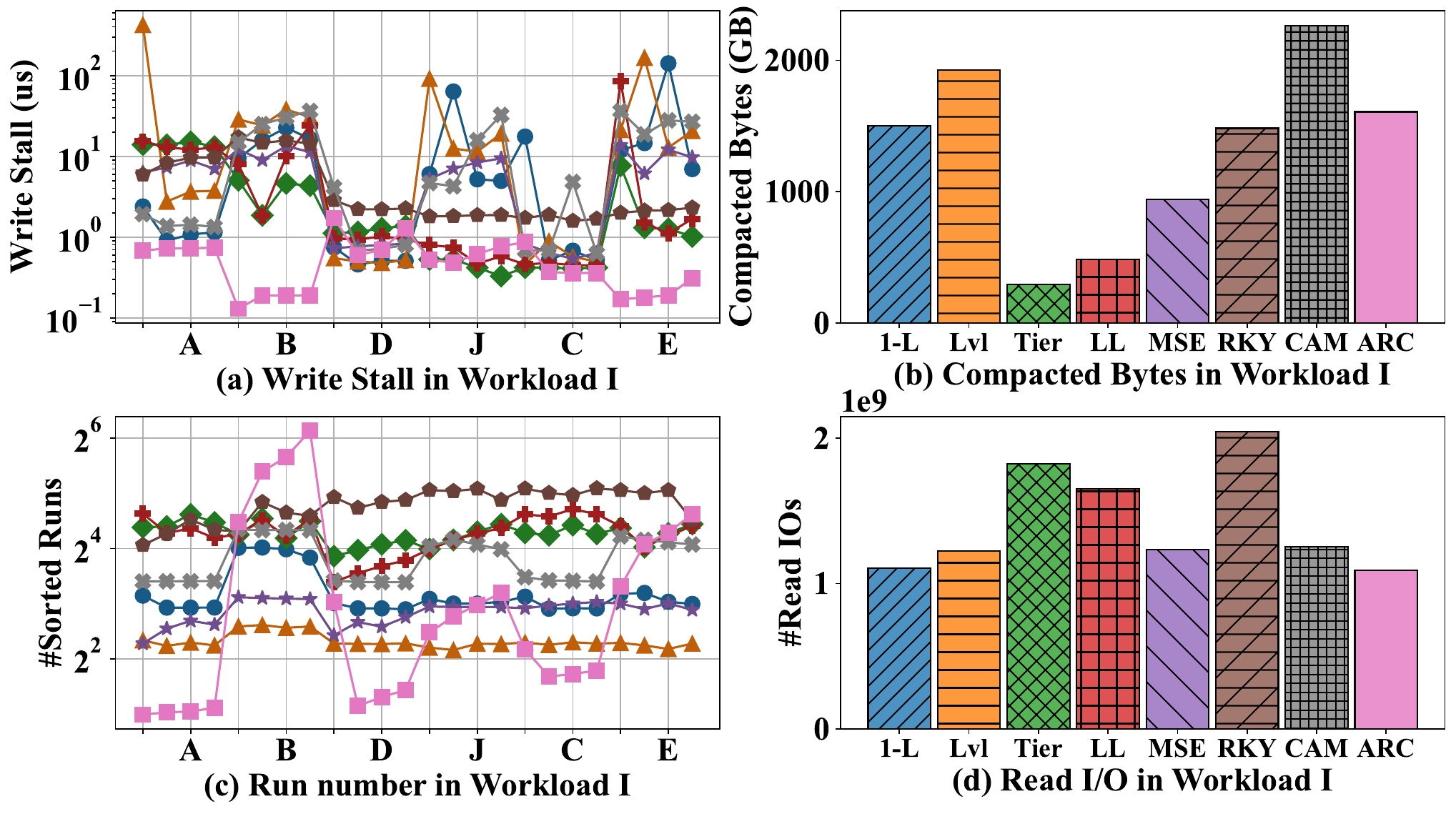}
\end{minipage}
\vspace{-18mm}
\caption{The table (left) shows the normalized throughput of each sub-workload in {\workloadfirst} and {\workloadsecond}; The figures (right) presents the change of write stall, total compaction bytes, the change of sorted runs, and the total read I/O when executing {\workloadfirst}.}
\vspace{-6mm}
\label{fig:workload_details}
\end{figure*}

\noindent{\bf Baselines.} The following systems and compaction strategies are used as baselines:
\begin{itemize}[leftmargin=*]
    \item {\bf Leveling (abbr. Lvl)}: Maintains at most one sorted run per level and increases level capacity using a fixed size ratio $T$. This policy is optimal for read-intensive workloads.
    \item {\bf Tiering (abbr. Tier)}: Allows up to $T$ sorted runs per level, also growing capacity by size ratio $T$. It is designed to favor write-intensive workloads by minimizing compaction overhead.
    \item {\bf 1-Leveling (abbr. 1-L)}: The default compaction style in RocksDB. Unlike traditional Leveling, it allows up to 20 sorted runs at the first level, making it particularly effective for read-heavy workloads and also adaptive to workloads with a portion of writes.
    \item {\bf LazyLeveling (abbr. LL)}: Structurally similar to Tiering but maintains only one sorted run at the largest level. This hybrid design improves performance for mixed read-write workloads.
    \item {\bf Moose (abbr. MSE)}: Leverages a dynamic programming algorithm to configure an LSM-tree structure that achieves an optimal balance among point lookups, range queries, and updates based on the given workload.
    \item {\bf Ruskey (abbr. RKY)}: Uses a reinforcement learning (RL) model to guide structural transitions, reducing the overhead of adapting to new workloads. Ruskey fixes the size ratio at $T=10$, while the number of sorted runs at each level is determined by the RL policy.
    {\item {\bf CAMAL (abbr. CAM)}: Employs active learning to optimize structural parameters, such as size ratio and the number of sorted runs per level as well as the allocation of memory to Bloom filter, cache, and Memtable.}
\end{itemize}
\noindent{\bf Implementation of Baselines.}
To ensure a fair comparison, all compaction policies except Ruskey\footnote{The implementation of Ruskey was obtained from the original authors. It is built on RocksDB.} and CAMAL\footnote{\url{https://github.com/NTU-Siqiang-Group/CAMAL}} are implemented on top of the Moose codebase\footnote{\url{https://github.com/NTU-Siqiang-Group/MooseLSM}}, a framework based on RocksDB that supports configurable numbers of sorted runs and level capacities. Leveling, Tiering, and LazyLeveling are implemented using this framework, each configured with a size ratio of $T=10$.
For Moose, the size ratios and the number of sorted runs per level are determined by its dynamic programming algorithm. In our evaluation, we set these values to 40, 40, 41 for size ratios and 6, 6, 6 for sorted runs, ensuring they accommodate the total number of entries to be ingested. {\bf {\ourkv} (abbr. ARC)} implements on top of RocksDB and sets the total level to 4, the I/O cost $I_w=15$us and $I_r=12$us based on the system profiling result.
We use a Bloom filter with 10 bits-per-key for all baselines, following RocksDB’s default implementation. Keys are fixed at 24 bytes, and values at 1,000 bytes and the write buffer is 2MB which is consistent with configurations commonly used in prior studies~\cite{ruskey,moose,dostoevsky2018}.

\noindent{\bf Dynamic Workload Generation.}
To construct realistic dynamic workloads, we primarily follow the approach used in Ruskey, which involves combining multiple sub-workloads with varying read-write ratios. To further enhance flexibility and realism, we distinguish between point lookups and range lookups when composing these workloads. We adopt several typical workload configurations evaluated in Endure~\cite{huynh2021endure} and Moose~\cite{moose}, with their compositions summarized in Table~\ref{tab:workload}. We define two primary compound workloads:
\begin{itemize}[leftmargin=*]
    \item {\bf \workloadfirst}: This compound workload consists of A, B, D, J, C, and E, with each sub-workload containing 40,960,000 operations, totaling 245,760,000 operations. It inserts approximately 74GiB of new data.
    \item {\bf \workloadsecond}: This workload comprises J, E, B, F, D, and C, with 40,960,000 operations per sub-workload, totaling 245,760,000 operations and inserting 94GiB new data.
\end{itemize}
The workload ratios in {\workloadfirst} shift more abruptly—from a high write ratio to a low one—while {\workloadsecond} exhibits a more gradual transition, with the write ratio slowly decreasing over time. 
For multi-threaded evaluation and comparison with industrial KV stores, we use {\workloadthird}, composed of G, H, and I workloads, each with 20,480,000 operations. This setup avoids extreme read- or write-heavy cases that are less reflective of production scenarios.
Before running any test workload, we preload the system by sequentially writing 40GiB of data (about 40 million entries) to ensure a stable compaction state, enabling fair evaluation of write-friendly strategies like Tiering and LazyLeveling.

\begin{figure*}
\vspace{-5mm}
\begin{minipage}[b]{.5\linewidth}
\includegraphics[width=\linewidth]{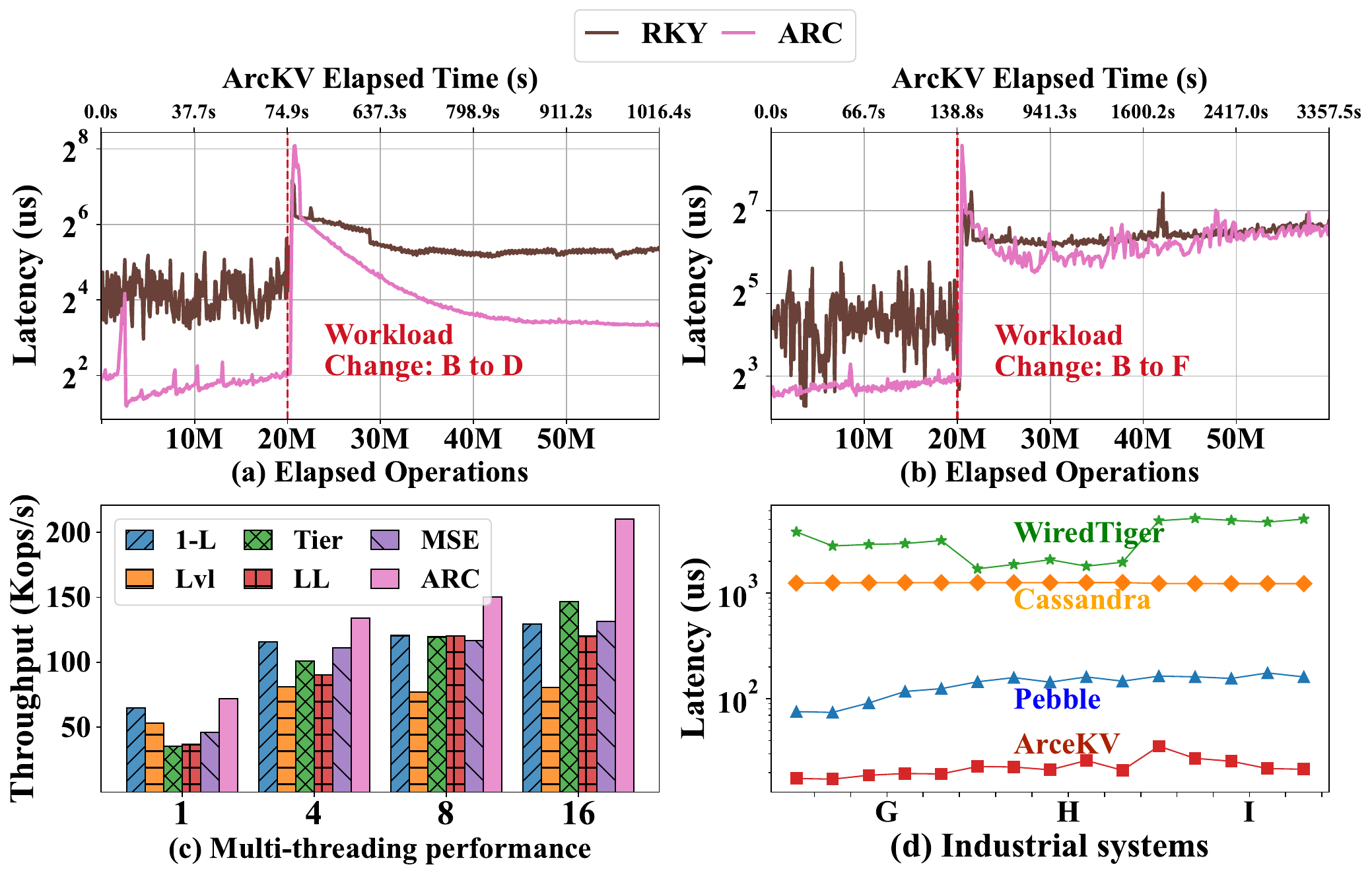}
\end{minipage}%
\begin{minipage}[b]{.5\linewidth}
\hspace{10mm}
\captionsetup{singlelinecheck = false, justification=raggedright, margin={5mm, 2mm}}
\captionof{table}{Performance comparison under diverse query distribution workloads in the YCSB benchmark.
}
\vspace{-3mm}
\raisebox{17.5mm}{
\centering
\scalebox{0.93}{
\renewcommand\arraystretch{1.2}
\begin{tabular}{|l|llllll|}
\hline
              & \multicolumn{6}{c|}{Normalized Throughput}                                                                                                                                        \\ \cline{2-7} 
              & \multicolumn{1}{c}{YCSB-A}  & \multicolumn{1}{c}{YCSB-B}  & \multicolumn{1}{c}{YCSB-C}  & \multicolumn{1}{c}{YCSB-D}  & \multicolumn{1}{c}{YCSB-E}  & \multicolumn{1}{c|}{YCSB-F} \\ \hline
1-L           & 1.44$\times$                & 1.00$\times$                & 1.94$\times$                & 1.00$\times$                & 3.89$\times$                & 1.42$\times$                \\
Lvl           & 1.00$\times$                & 1.25$\times$                & 1.53$\times$                & 1.13$\times$                & 2.79$\times$                & 1.00$\times$                \\
Tier          & 2.77$\times$                & 1.25$\times$                & 1.04$\times$                & 1.14$\times$                & 1.00$\times$                & 2.77$\times$                \\
LL            & 2.11$\times$                & 1.29$\times$                & 1.00$\times$                & 1.24$\times$                & 1.07$\times$                & 2.17$\times$                \\
MSE           & 1.87$\times$                & 1.84$\times$                & 1.62$\times$                & 1.57$\times$                & 2.89$\times$                & 1.89$\times$                \\
RKY           & 2.42$\times$                & 2.05$\times$                & 1.55$\times$                & 1.25$\times$                & 1.45$\times$                & 2.72$\times$                \\
CAM           & 1.37$\times$                & 2.03$\times$                & 1.49$\times$                & 1.60$\times$                & 2.45$\times$                & 1.37$\times$                \\
\textbf{OURS} & {\ul \textbf{4.18$\times$}} & {\ul \textbf{2.26$\times$}} & {\ul \textbf{2.34$\times$}} & {\ul \textbf{1.89$\times$}} & {\ul \textbf{5.47$\times$}} & {\ul \textbf{4.21$\times$}} \\ \hline
\end{tabular}
}}
\label{tab:ycsb}
\end{minipage}
\vspace{-8mm}
\caption{(a) and (b) illustrate the stabilization speed of {\ourkv}; (c) reports its performance under multi-threaded workloads; and (d) compares {\ourkv} with other industrial key–value stores. The table on the right summarizes {\ourkv}’s performance under skewed workloads (e.g., YCSB). Throughputs (operations per second) of all systems are normalized relative to the minimum value observed in each workload.}
\vspace{-6mm}
\label{fig:more_workloads}
\end{figure*}
\vspace{-3mm}
\subsection{System Performance}
\label{sec:sys_perf}

\noindent{\bf Overall, {\ourkv} demonstrates consistently strong performance under both {\workloadfirst} and {\workloadsecond}.}
Figure~\ref{fig:prim_workload_test} compares seven methods across two long-running workloads, {\workloadfirst} and {\workloadsecond}, each composed of multiple subworkloads with diverse read-write ratios, totaling up to 140GiB of data. Across both workloads, {\ourkv} consistently maintains top-tier performance, ranking first in nearly all subworkloads, except during the transition from workload B to workload D. This specific transition involves a sharp shift from an extremely write-intensive workload (B) to a read-intensive one (D), posing a significant challenge for systems to adapt promptly. Despite this abrupt change, {\ourkv} still performs competitively, only slightly trailing behind the two read-optimized baselines, 1-Leveling and Leveling.
In {\workloadsecond}, the transition from F to D is less drastic. Here, {\ourkv} performs similarly with 1-Leveling and Leveling. While these two baselines are optimized for read-heavy workloads, they struggle in write-intensive scenarios due to frequent and aggressive compactions. Conversely, Tiering and LazyLeveling allow more sorted runs per level and adopt lazier compaction strategies, reducing write stall and performing well under write-heavy workloads, but at the cost of degraded read performance due to excessive run accumulation.

Moose does not offer a robust solution for adapting structural configurations across varying workloads. The configurations it generates differ significantly (e.g., size ratios changing from {7,7,7,6,5} to {24,23,24}), making direct transitions between them impractical without incurring severe performance degradation. To mitigate this, we compute a unified configuration for Moose based on the total number of inserted entries and the average workload composition. Similarly, since CAMAL is also designed for static workloads, we provide its trained model with the average workload composition to derive the structural parameters. While both CAMAL and Moose perform well overall under this setting, their performance deteriorates significantly in extreme workloads such as B and E, falling well behind the top-performing baselines.
In contrast, Ruskey dynamically computes the most suitable configuration and employs an efficient adaptation strategy. While Ruskey performs satisfactorily in scenarios with gradual workload changes or sufficient updates, its adaptation can lag under abrupt shifts or under read-intensive workload. 
As shown in Figure~\ref{fig:workload_details}(a) and (b), Ruskey achieves a convergence rate comparable to {\ourkv} under workload F with sufficient updates, but adapts more slowly during transitions to highly read-intensive workloads such as D.

Furthermore, as shown in Figure~\ref{fig:prim_workload_test}(c) and (d), {\ourkv} maintains consistently low P99.9 latency over time, without incurring significant write stall or read overhead compared to other methods. This demonstrates the robustness and stability of {\ourkv} under varying workload conditions.

\noindent{\bf {\ourkv} orchestrates compactions and stalls more intelligently through the flexible LSM structure {\ourlsm}.}
Figure~\ref{fig:workload_details} presents additional metrics from {\workloadfirst} to highlight the benefits of {\ourlsm}'s structural flexibility. Under read-intensive workloads such as A and D, {\ourlsm} is able to reduce the number of sorted runs more quickly than even Leveling by performing multi-level compactions. In contrast, under write-intensive workloads like B, {\ourlsm} defers compactions more aggressively than Tiering and LazyLeveling by allowing more sorted runs to accumulate in the system. This enables {\ourkv} to maintain relatively low write stall time while still performing compactions as eagerly as Leveling when necessary.
As a result, {\ourkv} achieves nearly the same read I/O efficiency as 1-Leveling over time, demonstrating its ability to strike a dynamic balance between compaction aggressiveness and write stall control.
    
    
    
\begin{figure*}[t]
    \centering
    \vspace{-8mm}
    \includegraphics[width=\linewidth]{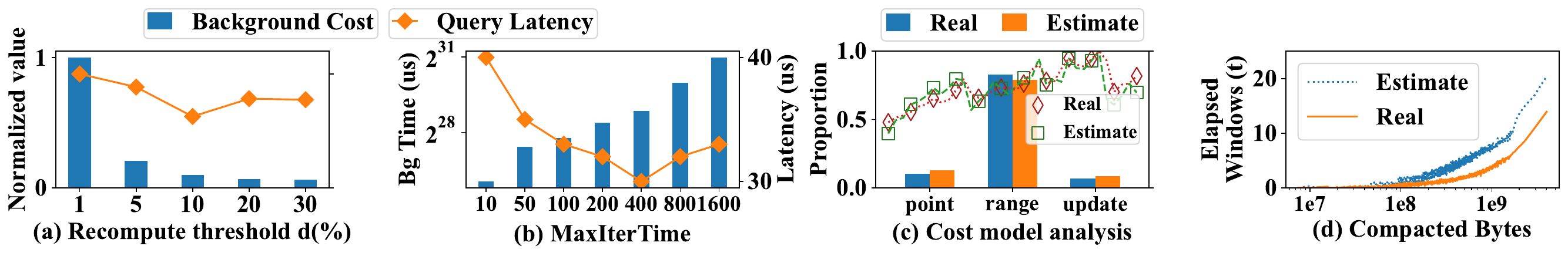}
    \vspace{-9mm}
    \caption{Internal parameter studies of {\ourkv}}
    \vspace{-3mm}
    \label{fig:internal_param}
\end{figure*}
\begin{figure*}[t]
    \centering
    \includegraphics[width=\linewidth]{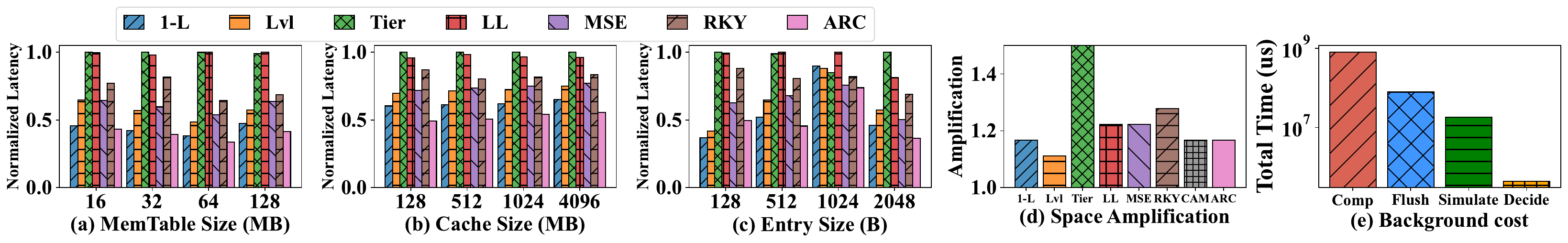}
    \vspace{-8mm}
    \caption{Common LSM parameter studies}
    \vspace{-4mm}
    \label{fig:common_param}
\end{figure*}
\begin{figure*}[t]
    \centering
    \includegraphics[width=\linewidth]{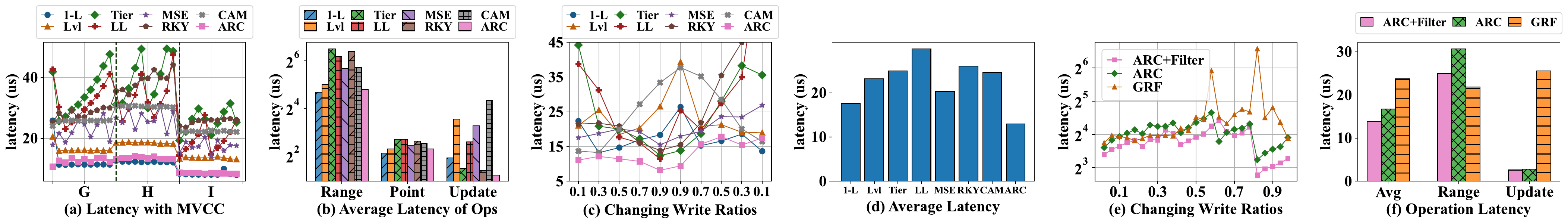}
    \vspace{-8mm}
    \caption{Case studies of various cases: (a) and (b) are the latencies of long multi-versioned chain scenario; (c) and (d) are latencies of continuous workloads; (e) and (f) are latency of {\ourkv} integrated with range filter.}
    \vspace{-6mm}
    \label{fig:case_studies}
\end{figure*}

\noindent{\bf {\ourkv} exhibits superior scalability under concurrent workloads than other baselines.}
As shown in Figure~\ref{fig:more_workloads}(c), we evaluate six methods (excluding Ruskey and CAMAL) using {\workloadthird} under 1, 4, 8, and 16 foreground threads. Ruskey and CAMAL are omitted because they rely on collecting performance metrics from the system and lacks synchronization mechanisms for multi-threaded scenarios. Additionally, the Moose framework does not support multiple background workers, so we fix the number of background threads to one and vary only the number of foreground query threads.
Overall, as the number of threads increases, all methods show improved throughput, but {\ourkv} achieves the greatest gains. This is because higher thread counts intensify updates, accelerating sorted run accumulation. Methods like 1-Leveling and Leveling struggle with frequent write stalls due to low stall thresholds, while Tiering and LazyLeveling perform better thanks to higher thresholds. In contrast, {\ourkv} dynamically adjusts its stall threshold and compaction scheduling (Section~\S\ref{sec:arcekv}) to account for multi-threading, effectively balancing read and write performance and delivering the highest scalability.

\noindent{\bf {\ourkv} outperforms other industrial key-value stores.}
As shown in Figure~\ref{fig:more_workloads}(d), {\ourkv} achieves the highest throughput when evaluated against Pebble, WiredTiger, and Cassandra on {\workloadthird}. WiredTiger offers strong write performance but suffers on range lookups due to its size-tiered compaction, which merges up to 15 chunks~\cite{Wiredtigerdoc}, leading to more sorted runs and higher lookup overhead. Cassandra incurs substantial local overhead from its distributed consistency mechanisms, resulting in lower performance. Pebble performs slightly worse than RocksDB, likely due to Go’s garbage collection and system call costs. In contrast, {\ourkv} efficiently manages on-disk data and dynamically adjusts write stalls, consistently outperforming these baselines.

\noindent{\bf {\ourkv} delivers strong performance under query-skewed workloads.}
To evaluate {\ourkv}'s ability to handle skewed query distributions, we benchmark it against six alternative methods using the YCSB suite. The workloads include: YCSB-A (Read-Write balanced), YCSB-B (Read-heavy), YCSB-C (Read-only), YCSB-D (Read-heavy with latest keys), YCSB-E (Range-heavy with latest keys), and YCSB-F (Read-Update balanced). Among them, workloads YCSB-C and YCSB-B follow a Zipfian distribution, YCSB-D focuses on the most recent insertions, while YCSB-A, YCSB-E, and YCSB-F use a uniform distribution.
Across all workloads, {\ourkv} consistently achieves top-tier performance.

{\ourkv} performs exceptionally well in write-read mixed workloads because it avoids re-simulating parameter configurations when the workload shifts, resulting in consistent and robust performance. Moreover, by broadening the compaction space through {\ourlsm}, {\ourkv} is able to identify more effective actions to optimize system behavior.


\vspace{-3mm}
\subsection{Parameter Studies}

\noindent{\bf Recomputing Threshold $d$}
Figure~\ref{fig:internal_param}(a) shows the background cost and foreground query latency under different values of the recomputing threshold $d$, which determines how much the tree state or workload must change before recomputing the parameter tuple $(M, c, k)$. When $d$ is too small, parameters are recomputed too frequently, resulting in substantial background overhead that may slightly impact foreground query latency. As $d$ increases, the background cost drops sharply and stabilizes beyond $d = 0.1$, where query latency also reaches its lowest point. This suggests that $d = 0.1$ offers a good balance between minimizing recomputation cost and maintaining responsiveness to dynamic workload changes.

\noindent{\bf Simulation \texttt{MaxIterTime}.}
This parameter controls the number of iterations (i.e., compaction decisions) used to evaluate the effectiveness of a given parameter tuple $(M, c, k)$ under the current workload and tree state in the simulation. As shown in Figure~\ref{fig:internal_param}(b), setting \texttt{MaxIterTime} too high can prolong the simulation and delay the timely adaptation of parameters, potentially degrading performance. Conversely, setting it too low may fail to capture the long-term effects of compactions, leading to suboptimal parameter choices. Through empirical evaluation, we find that 400 iterations offer a good balance between responsiveness and evaluation accuracy.

\noindent{\bf Justification of {\ourkv}'s Cost Model.}
To evaluate the accuracy of the cost model proposed in Section~\S\ref{sec:arce}, we compare the model’s predictions against actual system latency during the execution of sub-workload J in {\workloadsecond}. As shown in Figure~\ref{fig:internal_param}(c), the predicted cost closely follows the trend of the observed latency over time. For each operation, we decompose both the theoretical and real costs into three categories: point lookups, range lookups, and updates. The results indicate that the estimated cost for each operation type aligns well with the actual measured latency, validating the effectiveness of our model.

\noindent{\bf Accuracy of Estimating Compaction Windows.} In Section~\S\ref{sec:compaction_est}, we estimate the number of count windows spanned by a compaction using Equation~\ref{eq:t}. Figure~\ref{fig:internal_param}(d) presents over 500 estimated compaction durations, ranging from 4MB to nearly 20GB, compared against their actual elapsed windows when running sub-workload J in {\workloadsecond}. The results indicate that while estimation accuracy decreases slightly for larger compactions with longer execution times, the absolute error remains bounded within 3 windows.

\noindent{\bf Common Parameter Studies of LSM-Trees.}
To assess the robustness of our method, we evaluate {\ourkv} using {\workloadthird} under a variety of commonly used LSM-tree configurations. As these parameters are tuned by CAMAL's model, we exclude it from this evaluation. These include buffer sizes ranging from 16MB to 128MB, cache sizes from 128MB to 4096MB, and entry sizes from 128B to 2048B, presented in Figure~\ref{fig:common_param}(a) to (c). Across all 12 tested configurations, {\ourkv} consistently delivers the strongest or near-best performance, demonstrating its adaptability and resilience to diverse system settings.

\noindent{\bf Space Amplification.}
Although {\ourkv} focuses on read and write performance rather than space efficiency, its space amplification remains well controlled, as shown in Figure~\ref{fig:common_param}(d). Under workload J, which includes 36GiB of new updates and 18GiB of duplicates, Leveling yields the lowest amplification, but {\ourkv} follows closely with only a 0.05 increase. This is because {\ourkv}'s lookup-driven compactions still merge overlapping runs, effectively removing duplicates and limiting space growth.

\noindent{\bf Background Cost in {\ourkv}.}
Figure~\ref{fig:common_param}(e) shows the CPU time spent on background tasks in {\ourkv}, including compaction, flushing, simulation, and decision-making. With suitable values for $d$ and \texttt{MaxIterTime}, simulation adds only ~2\% overhead, and decision-making contributes less than 1\%, due to the efficiency of score-based selection and SIMD acceleration.

\vspace{-2mm}
\subsection{Case Studies}
\label{sec:case_studies}
\noindent{\bf Read Amplification under MVCC.} When updates concentrate on a small key set, write amplification stays low but long version chains can hurt read performance. To evaluate {\ourkv} under this condition, we simulate a multi-version workload with 1 million hot keys over 100 million operations, following Workload III (updates, point, and range lookups). We report average latency and read amplification, as requested by the reviewer, in Figure~\ref{fig:case_studies}(a) and (b).
In theory, Leveling and 1-Leveling are well suited to this setting due to the small size increase (1 GiB) and aggressive pruning. Similarly, {\ourlsm} detects small incoming runs, assigns low compaction penalties (Equation~\ref{eq:effect}), and converges to Leveling-like behavior, achieving comparable read latency and amplification.
\noindent{\bf Continuously Changing Workload.} We gradually vary the write ratio from 10\% to 90\% and back to 10\%, creating a smooth workload gradient rather than abrupt changes. At each ratio, we issue about 500,000 mixed read and write operations to reflect realistic fluctuations. As shown in Figure~\ref{fig:case_studies}(c) and (d), {\ourkv} achieves the best or near-best performance throughout execution and the lowest average latency. These results demonstrate that {\ourkv} maintains stable, adaptive performance under smoothly evolving workloads, enabled by its workload-aware compaction and write-stall mechanisms.

\noindent{\bf {\ourkv} with Range Filters.} As shown in Figure~\ref{fig:case_studies} (e) and (f), to assess how range filters~\cite{grafite,GRF,chen2024oasis} enhance {\ourkv}, we integrate Grafite~\cite{grafite}, a lightweight and easily pluggable range filter, and compare it with GRF~\cite{GRF}, a state-of-the-art LSM-tree system featuring range filtering. Grafite in {\ourkv} is serialized into SSTables during compaction and deserialized during queries. 
Additionally, we record the effective ratio of the filters (defined as the reduction in run accesses divided by the total number of run accesses) and multiply this ratio by the cost of range lookups.
We evaluate mixed workloads with update ratios from 10\% to 90\%. Overall, {\ourkv}+Grafite improves throughput by about 20\% over the original version and achieves roughly half the latency of GRF. The gap arises because GRF’s LazyLeveling policy improves point lookups and updates but adapts poorly to shifting update ratios, whereas its advanced position encoding gives it an edge in pure range lookups. Incorporating such encoding into {\ourkv} remains a direction for future work.

\vspace{-3mm}
\section{Related Work}
\label{sec:related}

\noindent{\bf Key-value stores.}
Over the past decade, extensive research has advanced key-value stores. Hardware-focused studies~\cite{xu2024ionia,duffy2023dotori,yan2021revisiting,lee2022dinomo,wang2014efficient,vinccon2018noftl,zhang2020fpga,benson2021viper} optimize for modern storage technologies, including advanced SSDs~\cite{yu2022treeline,duffy2023dotori,yan2021revisiting,xu2024ionia,chen2018co}, RDMA~\cite{wang2023replicating,wei2021xstore}, non-volatile memory~\cite{kaiyrakhmet2019slm,lee2022dinomo,zhang2023pm,zhong2023redesigning}, and disaggregated memory architectures~\cite{xiong2024distore,shen2023fusee,hu2024aceso}. These systems improve parallelism~\cite{wang2014efficient} and write throughput~\cite{vinccon2018noftl,zhang2020fpga,benson2021viper} by aligning architectures with hardware capabilities.
In cloud environments, several works~\cite{kejriwal2016slik,Idreos2019DesignCA,sivasubramanian2012amazon} integrate cloud-specific overheads into system design, while others rethink key-value store architectures entirely~\cite{lepers2019kvell,zhang2022depart,adya2019fast,wei2021xstore,zhu2024memory,mammarella2009modular}, enabling new system-level innovations.
Unlike our work, these efforts target diverse environments and architectures, rather than focusing specifically on LSM-tree optimization.

\noindent{\bf Optimization of LSM-based key-value stores.}
There has been rapid progress in optimizing LSM-trees, driven by rethinking their core components. Innovations include advanced compaction policies~\cite{chatterjee2021cosine,moose,dostoevsky2018,lsmbush,huynh2021endure,idreos2019design,sarkar2022constructing,huynh2024towards,vertiorizon,ecotune,zhao2025autumn} and update-friendly compaction schemes~\cite{raju2017pebblesdb,dayanspooky,sarkar2020lethe,sears2012blsm,wu2015lsmtrie,yao2017light,yao2017building}, both aiming to balance write amplification and responsiveness. Other work improves filtering structures like Bloom filters~\cite{dayan2017monkey,dayan2021chucky,li2022seesaw,qader2018comparative,zhang2018elasticbf,zhu2021reducing,Mnemosyne}, range filters~\cite{knorr2022proteus,luo2020rosetta,zhang2018surf,chen2024oasis}, and cache policies~\cite{wu2020ac,wang2024range} to cut unnecessary I/O.
Additional directions leverage emerging hardware~\cite{liang2021kvimr,ahmad2015compaction,thonangi2017log,vinccon2018noftl,wang2014efficient,zhang2020fpga,dflush}, bridge LSM-trees with update-in-place architectures~\cite{yu2022treeline}, and optimize via key-value separation~\cite{lu2017wisckey,dai2020wisckey,dai2020learning}, disaggregated storage~\cite{bindschaedler2020hailstorm}, selective flushing~\cite{balmau2017triad}, and improved memory/concurrency management~\cite{golan2015scaling,shetty2013building,bortnikov2018accordion,kim2020robust,luo2020breaking}. Other studies target tail-latency reduction~\cite{balmau2019silk,luo2019performance,sears2012blsm}, exploit data characteristics~\cite{absalyamov2018lightweight,ren2017slimdb,yang2020leaper,aster}, or adapt LSM-trees to cloud environments~\cite{polardb,chatterjee2021cosine,chatterjee2024limousine}. Most, however, optimize LSM-based key-value stores without explicitly incorporating workload characteristics into their design.




\noindent{\bf Cost-Model-Based Compaction Optimization.}
Unlike traditional Leveling and Tiering policies, recent LSM-tree optimization techniques employ cost models that quantify the cost of different operations to identify the most efficient compaction strategy~\cite{dostoevsky2018,lsmbush,ecotune,moose,vertiorizon,camal,ruskey}. These approaches can be broadly classified into two categories: workload-unaware and workload-aware methods. The former seek generally robust performance across diverse workloads, whereas the latter aim to achieve optimal efficiency for specific workload patterns.
Workload-unaware methods such as LazyLeveling~\cite{dostoevsky2018}, EcoTune~\cite{ecotune}, Moose~\cite{moose}, Spooky~\cite{dayanspooky}, and QLSM-Bush~\cite{lsmbush} model the costs of various operations and theoretically determine structural parameters or compaction behaviors that provide stable performance across diverse workloads. In contrast, workload-aware approaches such as Vertiorizon~\cite{vertiorizon}, CAMAL~\cite{camal}, Endure~\cite{huynh2021endure}, and SMoose~\cite{moose}, explicitly incorporate the operation probabilities of a given workload into their cost models, enabling more accurate parameter tuning for targeted scenarios. However, while workload awareness improves specialization, adapting between parameter configurations as workloads shift can incur substantial transition costs. Ruskey~\cite{ruskey} addresses this issue using reinforcement learning to mitigate transition inefficiencies. By contrast, {\ourkv} not only considers workload composition but also explicitly models and optimizes the {\it transition process} itself, enabling continuous performance optimization rather than discrete parameter transitioning.


\vspace{-3mm}
\section{Conclusion}
Existing LSM-based key-value stores perform poorly under dynamic workloads due to static compaction and write-stall policies. We propose {\ourkv}, a workload-driven system that adaptively schedules compactions and adjusts write-stall thresholds based on the current workload and LSM-tree state. Experiments show that {\ourkv} consistently outperforms state-of-the-art approaches under evolving workloads.
\vspace{-3mm}
\section{Acknowledgement}
This research is supported by the Ministry of Education, Singapore, under its AcRF Tier 2 programme (MOE-T2EP20224-0005), and NTU SUG-NAP (022029-00001). Any opinions, findings and conclusions or recommendations expressed in this material are those of the author(s) and do not reflect the views of the Ministry of Education, Singapore.
\section{Theoretical Insights}
\label{sec:proof-details}
\subsection{Proof of Lemma \ref{lmm:nphard}}
In our problem setting, the workloads are constantly evolving. Therefore, we can assume each workload only remains for a specific number of windows, and the $i$-th window corresponds to a workload composition $(r_i,u_i,p_i)$.


We consider the compaction sequence under a fixed number of windows $W$ by setting $r_i=u_i=p_i=0$ for $i>W$, which means the compaction sequence should span $W$ windows. Note that a ``no compaction'' strategy can also be considered as a compaction, which spans 1 window but does not reduce any sorted runs. Thus, the denominator part of the average cost is fixed and we only consider the numerator part (i.e., total cost). Now, we reduce the decision version of our new problem from the NP-complete Equal-Cardinality Partition problem \cite{partitionprob}.

The equal-cardinality partition problem is described as follows: given a multiset $\left\{s_1,s_2,\cdots,s_{2n}\right\}$, determine whether there exists a disjoint partition into two subsets $\left\{s_{i_1},s_{i_2},\cdots,s_{i_n}\right\}$ and $\left\{s_{j_1},s_{j_2},\cdots,s_{j_n}\right\}$, such that their sums and sizes are equal.

Let the initial tree state $S=\left\{s_1,\cdots,s_{2n}\right\}$ be the target set in the partition problem. Set the maximum number of layers of the LSM-tree as 1, such that all sorted runs lie within the same layer and we can compact any subset of them in $S$ (i.e., intra-level compaction). Set $I_w$ to be a sufficiently small number and $I_r=1$, such that the I/O cost is determined only by lookups. Set $W=3$, and then there are two possible cases for an optimal compaction sequence:
\begin{enumerate}[leftmargin=*]
    \item Perform a compaction that spans the first window, another compaction that spans the second window, and a ``no compaction'' that spans the third window;
    \item Perform a compaction that spans the first and second windows, and a ``no compaction'' that spans the third window.
\end{enumerate}

We would show as follows that by appropriately setting $B,r_i,p_i$ and $\alpha$, the optimal compaction sequence lies within case 1 and has a specific upper bound in its objective function if and only if a valid partition exists.

We now analyze the costs of compactions of case 1. For case 1, the cost of the first compaction is:
\begin{align}
    r_1\cdot 2n+p_1\cdot(\alpha\cdot 2n+1)
\end{align}
and should be no less than $\frac{X_1}{B}$. Similarly, the cost of the second compaction is:
\begin{align}
    r_2\cdot (2n-y_1+1)+p_2\cdot (\alpha\cdot (2n-y_1+1)+1)
\end{align}
and should be no less than $\frac{X_2}{B}$. The cost of the third ``no compaction'' is:
\begin{align}
    r_3\cdot(2n-y_1-y_2+2)+p_3\cdot(\alpha\cdot(2n-y_1-y_2+2)+1)
\end{align}

If we set $B,r_i,p_i$ and $\alpha$ appropriately, such that:
\begin{equation}
  \begin{cases}
    r_1\cdot 2n+p_1\cdot(\alpha\cdot 2n+1)=\frac{T}{2B}
    \\r_2\cdot (n+2)+p_2\cdot (\alpha\cdot (n+2)+1)=\frac{T+2F}{2B}
\end{cases}  
\label{eq:condition}
\end{equation}
where $T=\sum_{i=1}^{2n}s_i$ is the sum of the set $S$. Then it is straightforward to check there is a valid compaction sequence with $X_1=\frac{T}{2},y_1=n-1,X_2=\frac{T+2F}{2},y_2=n$ and $X_3=0,y_3=0$ if a valid partition exists.

\noindent\textbf{Claim.} By setting $B,r_i,p_i$ and $\alpha$ to satisfy Equation \ref{eq:condition}, the average cost is at most $\frac{\frac{T+F}{B}+3r_3+p_3\cdot(3\alpha+1)}{3}$ if and only if a valid partition exists.
\begin{proof}
    If a valid partition exists, there is a valid compaction sequence with $X_1=\frac{T}{2},y_1=n-1,X_2=\frac{T+2F}{2},y_2=n,X_3=0,y_3=0$ as discussed above, which derives the desired results.

    Assume that a valid partition does not exist. We prove that the compaction costs exceed our desired results. Firstly consider a compaction sequence under case 1. For simplicity, we denote $C=3r_3+p_3\cdot(3\alpha+1)$. The first compaction must satisfy $X_1\leq\frac{T}{2}$ in order to finish within the first window, and consider cases for $y_1$:
    \begin{itemize}[leftmargin=*]
        \item $y_1<n-1$: the compaction costs are at least:
        $$\frac{T}{2B}+\textcolor{red}{r_2\cdot(2n-y_1+1)+p_2\cdot(\alpha\cdot(2n-y_1+1)+1)}+C$$
        where the middle term (in red) is larger than $\frac{T+2F}{2B}$, yielding a larger result;
        \item $y_1=n-1$: then $X_1<\frac{T}{2}$ since there does not exist a valid partition, and $X_2\leq\frac{T+2F}{2}$ such that the second compaction can finish within one window. However, these two conditions imply that $X_1+X_2<T+F$, which means the compaction sequence does not involve all elements in $S$ and the newly inserted sorted run in the first window. Then $y_1+y_2<2n-1$ and the compaction costs are at least:
        $$\frac{T+F}{B}+\textcolor{red}{(2n-y_1-y_2+2)r_3+(\alpha\cdot(2n-y_1-y_2+2)+1)p_3}$$
        where the latter term (in red) is larger than $C$, yielding a larger result;
        \item $y_1>n-1$: this implies $X_2<\frac{T+2F}{2}$ such that the second compaction can finish within one window. Then $X_1+X_2< T+F$, which implies $y_1+y_2<2n-1$, and this is similar to the case when $y_1=n-1$.
    \end{itemize}
    Next consider a compaction sequence under case 2, which implies the first compaction satisfies $X_1>\frac{T}{2}$. However, the compaction costs are at least:
    $$\frac{T}{2B}+\textcolor{red}{r_2\cdot(2n+1)+p_2\cdot(\alpha\cdot(2n+1)+1)}+C$$
    where the middle term (in red) is larger than $\frac{T+2F}{2B}$, yielding a larger result.
\end{proof}

\subsection{Technical Insights of Dominating Compaction}
In Section~\ref{sec:arce}, we propose pruning the large compaction action space by restricting attention to dominating compactions. In what follows, we establish two key results: (1) under this restriction, dominating compactions must be at least partially involved; otherwise, an optimal average cost cannot be achieved; and (2) for each dominating compaction, there exists a corresponding set of hyperparameters $(M,k,c)$ under which that compaction will be selected whenever the hyperparameters are appropriately configured.

\vspace{1mm}
\noindent\textbf{Claim 1.} If $A \prec B$, $A$ appears in the optimal compaction sequence $C$, and $B$ is non-dominated but none of $B$’s runs occur in $C$, then there exists a compaction sequence that includes $B$ and has an average cost no greater than that of $C$.

Suppose such a situation occurs. We replace $A$ with $B$ in the sequence. Since $t_A > t_B$, we wait for $t_A - t_B$ windows after executing $B$ to ensure that the newly flushed runs remain consistent with the original sequence. Let $s_A$ denote the resulting run produced by compaction $A$. For each subsequent compaction $C_i$ in the original sequence, if $C_i$ does not involve $s_A$, we keep its runs unchanged. Otherwise, we replace $s_A$ with the original runs of $A$. Because the reduced runs of $B$ are no smaller than those of $A$, the cost and the number of windows consumed by each updated $C_i$ will not increase. If an updated $C_i$ consumes fewer windows in the new sequence, we insert an appropriate waiting period to keep it aligned with the original timeline. Finally, since the cost of $B$ is smaller than that of $A$, the resulting sequence has an average cost no greater than that of the original sequence, and in fact strictly smaller when the replacement occurs. This completes the argument.

\vspace{1mm}
\noindent\textbf{Claim 2.} For each non-dominated compaction, there exist parameters $(M,c,k)$ such that its score is the highest.
We prove it by claiming that for each non-dominated compaction upon the decision point, there exist some $(M_i,k_i,c_i)$ such that its effectiveness score is the highest. To achieve this, assume the compaction reduces $y$ runs and costs $t$ time, $(M_i,k_i,c_i)$ should satisfy:
$$M_i\cdot a\cdot (y-y')-a\cdot (t-t')-uk_i[\max(0,t+s-c_i)-\max(0,t'+s-c_i)]\geq 0$$
for all $(y-y')(t-t')>0$, where $a=r+\alpha\cdot p$ is a constant.

We first analyze the case $y>y'$ and $t>t'$: in the worst case, $y'=y-1$ and $t'=1$, then:
$$
M_i\geq\begin{cases}
    \frac{(a+uk_i)(t-1)}{a}&\quad 0\leq c_i\leq s+1,\\
    t-1+\frac{uk_i}{a}(t+s-c_i) &\quad s+1<c_i\leq s+t,\\
    t-1 &\quad otherwise
\end{cases}
$$

For the case $y<y'$ and $t<t'$: in the worst case, $y'=s-1$ and $t'=t+1$, then:
$$
M_i\leq\begin{cases}
    \frac{a+uk_i}{a(s-y-1)}&\quad 0\leq c_i\leq s+t,\\
    \frac{1}{s-y-1} &\quad otherwise
\end{cases}
$$

Therefore, by setting $c_i=s+t,k_i=\frac{a(s-y-1)}{u}(t-1)$ and $M_i=t-1$, the compaction has the highest effectiveness score.

\bibliographystyle{ACM-Reference-Format}
\bibliography{sample-base}

@String{Computing = "Computing" }

@String{Computer = "{IEEE} Computer" }

@String{Springer = "Springer-Verlag" }

@inproceedings{dostoevsky2018,
author = {Dayan, Niv and Idreos, Stratos},
title = {Dostoevsky: Better Space-Time Trade-Offs for LSM-Tree Based Key-Value Stores via Adaptive Removal of Superfluous Merging},
year = {2018},
isbn = {9781450347037},
publisher = {Association for Computing Machinery},
address = {New York, NY, USA},
url = {https://doi.org/10.1145/3183713.3196927},
doi = {10.1145/3183713.3196927},
booktitle = {Proceedings of the 2018 International Conference on Management of Data},
pages = {505–520},
numpages = {16},
location = {Houston, TX, USA},
series = {SIGMOD '18}
}

@article{luo2020breaking,
  title={Breaking down memory walls: adaptive memory management in LSM-based storage systems},
  author={Luo, Chen and Carey, Michael J},
  journal={Proceedings of the VLDB Endowment},
  volume={14},
  number={3},
  pages={241--254},
  year={2020},
  publisher={VLDB Endowment}
}

@article{corbett2013spanner,
  title={Spanner: Google’s globally distributed database},
  author={Corbett, James C and Dean, Jeffrey and Epstein, Michael and Fikes, Andrew and Frost, Christopher and Furman, Jeffrey John and Ghemawat, Sanjay and Gubarev, Andrey and Heiser, Christopher and Hochschild, Peter and others},
  journal={ACM Transactions on Computer Systems (TOCS)},
  volume={31},
  number={3},
  pages={1--22},
  year={2013},
  publisher={ACM New York, NY, USA}
}

@misc{leveldb,
  howpublished = {\url{https://github.com/google/leveldb/} },
  title       = {LevelDB},
  year        = {2024},
  author = {Google}
}

@misc{rocksdb,
  howpublished = {\url{https://github.com/facebook/rocksdb} },
  title       = {RocksDB},
  year        = {2024},
  author = {Facebook}
}

@misc{tikv,
    howpublished = {\url{https://docs.pingcap.com/tidb/stable/tikv-configuration-file/} },
    title = {TiKV Tuning Guide},
    year = {2025},
    author = {PingCAP}
}

@misc{scylladb,
  howpublished = {\url{https://www.scylladb.com/} },
  title       = {ScyllaDB},
  year        = {2024},
  author = {Cloudius Systems}
}

@misc{cockroachdb,
  howpublished = {\url{https://github.com/cockroachdb/cockroach}},
  title       = {CockroachDB},
  year        = {2024},
  author = {Cockroach Labs}
}

@inproceedings{zhao2025autumn,
  title={Autumn: A Scalable Read Optimized LSM-Tree Based Key-Value Stores with Fast Point and Range Reads},
  author={Zhao, Fuheng and Miller, Zach and Reznikov, Leron and Agrawal, Divyakant and El Abbadi, Amr},
  booktitle={2025 IEEE 41st International Conference on Data Engineering (ICDE)},
  pages={2824--2837},
  year={2025},
  organization={IEEE Computer Society}
}

@misc{eigen,
  author       = {Guennebaud, Ga{\"e}l and Jacob, Beno{\^i}t and others},
  title        = {Eigen v3},
  howpublished = {\url{https://eigen.tuxfamily.org}},
  year         = {2010}
}

@article{tidb,
  title={TiDB: a Raft-based HTAP database},
  author={Huang, Dongxu and Liu, Qi and Cui, Qiu and Fang, Zhuhe and Ma, Xiaoyu and Xu, Fei and Shen, Li and Tang, Liu and Zhou, Yuxing and Huang, Menglong and others},
  journal={Proceedings of the VLDB Endowment},
  volume={13},
  number={12},
  pages={3072--3084},
  year={2020},
  publisher={VLDB Endowment}
}

@misc{dgraph,
  howpublished = {\url{https://dgraph.io/} },
  title       = {DGraph},
  year        = {2024},
  author = {DGraph}
}

@article{disco,
author = {Zhong, Wenshao and Chen, Chen and Wu, Xingbo and Eriksson, Jakob},
title = {Disco: A Compact Index for LSM-trees},
year = {2025},
issue_date = {February 2025},
publisher = {Association for Computing Machinery},
address = {New York, NY, USA},
volume = {3},
number = {1},
url = {https://doi.org/10.1145/3709683},
doi = {10.1145/3709683},
journal = {Proc. ACM Manag. Data},
month = feb,
articleno = {33},
numpages = {27}
}

@inproceedings{polardb,
  title={X-Engine: An optimized storage engine for large-scale E-commerce transaction processing},
  author={Huang, Gui and Cheng, Xuntao and Wang, Jianying and Wang, Yujie and He, Dengcheng and Zhang, Tieying and Li, Feifei and Wang, Sheng and Cao, Wei and Li, Qiang},
  booktitle={Proceedings of the 2019 International Conference on Management of Data},
  pages={651--665},
  year={2019}
}

@article{aster,
author = {Mo, Dingheng and Liu, Junfeng and Wang, Fan and Luo, Siqiang},
title = {Aster: Enhancing LSM-structures for Scalable Graph Database},
year = {2025},
issue_date = {February 2025},
publisher = {Association for Computing Machinery},
address = {New York, NY, USA},
volume = {3},
number = {1},
url = {https://doi.org/10.1145/3709662},
doi = {10.1145/3709662},
journal = {Proc. ACM Manag. Data},
month = feb,
articleno = {12},
numpages = {26}
}

@article{Mnemosyne,
author = {Zhu, Zichen and Wei, Yanpeng and Mun, Ju Hyoung and Athanassoulis, Manos},
title = {Mnemosyne: Dynamic Workload-Aware BF Tuning via Accurate Statistics in LSM trees},
year = {2025},
issue_date = {June 2025},
publisher = {Association for Computing Machinery},
address = {New York, NY, USA},
volume = {3},
number = {3},
url = {https://doi.org/10.1145/3725327},
doi = {10.1145/3725327},
journal = {Proc. ACM Manag. Data},
month = jun,
articleno = {190},
numpages = {28}
}

@article{dflush,
author = {Ding, Chen and Lu, Kai and Zhang, Quanyi and Ye, Zekun and Yao, Ting and Wang, Daohui and Wu, Huatao and Wan, Jiguang},
title = {DFlush: DPU-Offloaded Flush for Disaggregated LSM-based Key-Value Stores},
year = {2025},
issue_date = {June 2025},
publisher = {Association for Computing Machinery},
address = {New York, NY, USA},
volume = {3},
number = {3},
url = {https://doi.org/10.1145/3725284},
doi = {10.1145/3725284},
journal = {Proc. ACM Manag. Data},
month = jun,
articleno = {147},
numpages = {28}
}

@article{ecotune,
author = {Wang, Hengrui and Qiu, Jiansheng and Yuan, Fangzhou and Zhang, Huanchen},
title = {Rethinking The Compaction Policies in LSM-trees},
year = {2025},
issue_date = {June 2025},
publisher = {Association for Computing Machinery},
address = {New York, NY, USA},
volume = {3},
number = {3},
url = {https://doi.org/10.1145/3725344},
doi = {10.1145/3725344},
journal = {Proc. ACM Manag. Data},
month = jun,
articleno = {207},
numpages = {26}
}

@article{vertiorizon,
author = {Mo, Dingheng and Luo, Siqiang and Idreos, Stratos},
title = {How to Grow an LSM-tree? Towards Bridging the Gap Between Theory and Practice},
year = {2025},
issue_date = {June 2025},
publisher = {Association for Computing Machinery},
address = {New York, NY, USA},
volume = {3},
number = {3},
url = {https://doi.org/10.1145/3725310},
doi = {10.1145/3725310},
journal = {Proc. ACM Manag. Data},
month = jun,
articleno = {173},
numpages = {25}
}

@misc{wiredtiger,
  howpublished = {\url{https://github.com/wiredtiger/wiredtiger} },
  title       = {WiredTiger},
  year        = {2024},
  author = {Source Code}
}

@misc{influxdb,
  howpublished = {\url{https://www.influxdata.com/} },
  title       = {InfluxDB},
  year        = {2024},
  author = {Influxdata}
}

@article{moose,
  title={Structural Designs Meet Optimality: Exploring Optimized LSM-tree Structures in A Colossal Configuration Space},
  author={Liu, Junfeng and Wang, Fan and Mo, Dingheng and Luo, Siqiang},
  journal={Proceedings of the ACM on Management of Data},
  volume={2},
  number={3},
  pages={1--26},
  year={2024},
  publisher={ACM New York, NY, USA}
}

@inproceedings{idreos2019design,
  title={Design Continuums and the Path Toward Self-Designing Key-Value Stores that Know and Learn.},
  author={Idreos, Stratos and Dayan, Niv and Qin, Wilson and Akmanalp, Mali and Hilgard, Sophie and Ross, Andrew and Lennon, James and Jain, Varun and Gupta, Harshita and Li, David and others},
  booktitle={CIDR},
  year={2019}
}

@article{chatterjee2024limousine,
  title={Limousine: Blending Learned and Classical Indexes to Self-Design Larger-than-Memory Cloud Storage Engines},
  author={Chatterjee, Subarna and Pekala, Mark F and Kruglyak, Lev and Idreos, Stratos},
  journal={Proceedings of the ACM on Management of Data},
  volume={2},
  number={1},
  pages={1--28},
  year={2024},
  publisher={ACM New York, NY, USA}
}

@article{ruskey,
author = {Mo, Dingheng and Chen, Fanchao and Luo, Siqiang and Shan, Caihua},
title = {Learning to Optimize LSM-trees: Towards A Reinforcement Learning based Key-Value Store for Dynamic Workloads},
year = {2023},
issue_date = {September 2023},
publisher = {Association for Computing Machinery},
address = {New York, NY, USA},
volume = {1},
number = {3},
url = {https://doi.org/10.1145/3617333},
doi = {10.1145/3617333},
journal = {Proc. ACM Manag. Data},
month = nov,
articleno = {213},
numpages = {25}
}

@inproceedings{Idreos2019DesignCA,
  title={Design Continuums and the Path Toward Self-Designing Key-Value Stores that Know and Learn.},
  author={Idreos, Stratos and Dayan, Niv and Qin, Wilson and Akmanalp, Mali and Hilgard, Sophie and Ross, Andrew and Lennon, James and Jain, Varun and Gupta, Harshita and Li, David and others},
  booktitle={CIDR},
  year={2019}
}

@inproceedings{lsmbush,
  title={The log-structured merge-bush \& the wacky continuum},
  author={Dayan, Niv and Idreos, Stratos},
  booktitle={Proceedings of the 2019 International Conference on Management of Data},
  pages={449--466},
  year={2019}
}

@article{chatterjee2021cosine,
  title={Cosine: a cloud-cost optimized self-designing key-value storage engine},
  author={Chatterjee, Subarna and Jagadeesan, Meena and Qin, Wilson and Idreos, Stratos},
  journal={Proceedings of the VLDB Endowment},
  volume={15},
  number={1},
  pages={112--126},
  year={2021},
  publisher={VLDB Endowment}
}

@inproceedings{dayan2017monkey,
  title={Monkey: Optimal navigable key-value store},
  author={Dayan, Niv and Athanassoulis, Manos and Idreos, Stratos},
  booktitle={Proceedings of the 2017 ACM International Conference on Management of Data},
  pages={79--94},
  year={2017}
}

@article{lakshman2010cassandra,
  title={Cassandra: a decentralized structured storage system},
  author={Lakshman, Avinash and Malik, Prashant},
  journal={ACM SIGOPS Operating Systems Review},
  volume={44},
  number={2},
  pages={35--40},
  year={2010},
  publisher={ACM New York, NY, USA}
}

@article{huynh2021endure,
  title={Endure: a robust tuning paradigm for LSM trees under workload uncertainty},
  author={Huynh, Andy and Chaudhari, Harshal A and Terzi, Evimaria and Athanassoulis, Manos},
  journal={Proceedings of the VLDB Endowment},
  volume={15},
  number={8},
  pages={1605--1618},
  year={2022},
  publisher={VLDB Endowment}
}

@inproceedings{xu2024ionia,
  title={IONIA:High-Performance Replication for Modern Disk-based KV Stores},
  author={Xu, Yi and Zhu, Henry and Pandey, Prashant and Conway, Alex and Johnson, Rob and Ganesan, Aishwarya and Alagappan, Ramnatthan},
  booktitle={22nd USENIX Conference on File and Storage Technologies (FAST 24)},
  pages={225--241},
  year={2024}
}

@inproceedings{lepers2019kvell,
  title={Kvell: the design and implementation of a fast persistent key-value store},
  author={Lepers, Baptiste and Balmau, Oana and Gupta, Karan and Zwaenepoel, Willy},
  booktitle={Proceedings of the 27th ACM Symposium on Operating Systems Principles},
  pages={447--461},
  year={2019}
}

@inproceedings{liang2021kvimr,
  title={KVIMR:Key-Value Store Aware Data Management Middleware for Interlaced Magnetic Recording Based Hard Disk Drive},
  author={Liang, Yuhong and Yang, Tsun-Yu and Yang, Ming-Chang},
  booktitle={2021 USENIX Annual Technical Conference (USENIX ATC 21)},
  pages={657--671},
  year={2021}
}

@article{chen2018co,
  title={Co-optimizing storage space utilization and performance for key-value solid state drives},
  author={Chen, Yen-Ting and Yang, Ming-Chang and Chang, Yuan-Hao and Chen, Tseng-Yi and Wei, Hsin-Wen and Shih, Wei-Kuan},
  journal={IEEE Transactions on Computer-Aided Design of Integrated Circuits and Systems},
  volume={38},
  number={1},
  pages={29--42},
  year={2018},
  publisher={IEEE}
}

@inproceedings{bindschaedler2020hailstorm,
  title={Hailstorm: Disaggregated compute and storage for distributed lsm-based databases},
  author={Bindschaedler, Laurent and Goel, Ashvin and Zwaenepoel, Willy},
  booktitle={Proceedings of the Twenty-Fifth International Conference on Architectural Support for Programming Languages and Operating Systems},
  pages={301--316},
  year={2020}
}

@article{dai2020learning,
  title={Learning How To Learn Within An LSM-based Key-Value Store.},
  author={Dai, Yifan and Xu, Yien and Ganesan, Aishwarya and Alagappan, Ramnatthan and Kroth, Brian and Arpaci-Dusseau, Andrea C and Arpaci-Dusseau, Remzi H},
  journal={CoRR},
  year={2020}
}

@inproceedings{zhu2021reducing,
  title={Reducing bloom filter cpu overhead in lsm-trees on modern storage devices},
  author={Zhu, Zichen and Mun, Ju Hyoung and Raman, Aneesh and Athanassoulis, Manos},
  booktitle={Proceedings of the 17th International Workshop on Data Management on New Hardware (DaMoN 2021)},
  pages={1--10},
  year={2021}
}

@article{GRF,
author = {Wang, Hengrui and Guo, Te and Yang, Junzhao and Zhang, Huanchen},
title = {GRF: A Global Range Filter for LSM-Trees with Shape Encoding},
year = {2024},
issue_date = {June 2024},
publisher = {Association for Computing Machinery},
address = {New York, NY, USA},
volume = {2},
number = {3},
url = {https://doi.org/10.1145/3654944},
doi = {10.1145/3654944},
journal = {Proc. ACM Manag. Data},
month = may,
articleno = {141},
numpages = {27}
}

@article{grafite,
author = {Costa, Marco and Ferragina, Paolo and Vinciguerra, Giorgio},
title = {Grafite: Taming Adversarial Queries with Optimal Range Filters},
year = {2024},
issue_date = {February 2024},
publisher = {Association for Computing Machinery},
address = {New York, NY, USA},
volume = {2},
number = {1},
url = {https://doi.org/10.1145/3639258},
doi = {10.1145/3639258},
journal = {Proc. ACM Manag. Data},
month = mar,
articleno = {3},
numpages = {23}
}

@article{huynh2024towards,
  title={Towards flexibility and robustness of LSM trees},
  author={Huynh, Andy and Chaudhari, Harshal A and Terzi, Evimaria and Athanassoulis, Manos},
  journal={The VLDB Journal},
  pages={1--24},
  year={2024},
  publisher={Springer}
}

@inproceedings{dai2020wisckey,
  title={From WiscKey to Bourbon: A Learned Index for Log-Structured Merge Trees},
  author={Dai, Yifan and Xu, Yien and Ganesan, Aishwarya and Alagappan, Ramnatthan and Kroth, Brian and Arpaci-Dusseau, Andrea and Arpaci-Dusseau, Remzi},
  booktitle={14th USENIX Symposium on Operating Systems Design and Implementation (OSDI 20)},
  pages={155--171},
  year={2020}
}

@inproceedings{wu2015lsmtrie,
  title={LSM-trie: An LSM-tree-based Ultra-Large Key-Value Store for Small Data Items},
  author={Wu, Xingbo and Xu, Yuehai and Shao, Zili and Jiang, Song},
  booktitle={2015 USENIX Annual Technical Conference (USENIX ATC 15)},
  pages={71--82},
  year={2015}
}

@inproceedings{raju2017pebblesdb,
  title={Pebblesdb: Building key-value stores using fragmented log-structured merge trees},
  author={Raju, Pandian and Kadekodi, Rohan and Chidambaram, Vijay and Abraham, Ittai},
  booktitle={Proceedings of the 26th Symposium on Operating Systems Principles},
  pages={497--514},
  year={2017}
}

@article{yan2021revisiting,
  title={Revisiting the design of LSM-tree Based OLTP storage engine with persistent memory},
  author={Yan, Baoyue and Cheng, Xuntao and Jiang, Bo and Chen, Shibin and Shang, Canfang and Wang, Jianying and Huang, Gui and Yang, Xinjun and Cao, Wei and Li, Feifei},
  journal={Proceedings of the VLDB Endowment},
  volume={14},
  number={10},
  pages={1872--1885},
  year={2021},
  publisher={VLDB Endowment}
}

@inproceedings{sears2012blsm,
  title={bLSM: a general purpose log structured merge tree},
  author={Sears, Russell and Ramakrishnan, Raghu},
  booktitle={Proceedings of the 2012 ACM SIGMOD International Conference on Management of Data},
  pages={217--228},
  year={2012}
}

@article{duffy2023dotori,
  title={Dotori: A Key-Value SSD Based KV Store},
  author={Duffy, Carl and Shim, Jaehoon and Kim, Sang-Hoon and Kim, Jin-Soo},
  journal={Proceedings of the VLDB Endowment},
  volume={16},
  number={6},
  pages={1560--1572},
  year={2023},
  publisher={VLDB Endowment}
}

@inproceedings{cao2020characterizing,
  title={Characterizing, Modeling, and Benchmarking RocksDB Key-Value Workloads at Facebook},
  author={Cao, Zhichao and Dong, Siying and Vemuri, Sagar and Du, David HC},
  booktitle={18th USENIX Conference on File and Storage Technologies (FAST 20)},
  pages={209--223},
  year={2020}
}

@article{lee2022dinomo,
  title={DINOMO: an elastic, scalable, high-performance key-value store for disaggregated persistent memory},
  author={Lee, Sekwon and Ponnapalli, Soujanya and Singhal, Sharad and Aguilera, Marcos K and Keeton, Kimberly and Chidambaram, Vijay},
  journal={Proceedings of the VLDB Endowment},
  volume={15},
  number={13},
  pages={4023--4037},
  year={2022},
  publisher={VLDB Endowment}
}

@misc{netflix,
  howpublished = {\url{https://www.youtube.com/watch?v=n_SXhW-x0WA}},
  title = {How Netflix optimizes use of Apache Cassandra® for massive scale},
  year = {2024},
  author = {Netflix}
}

@article{dayanspooky,
  title={Spooky: granulating LSM-tree compactions correctly},
  author={Dayan, Niv and Weiss, Tamar and Dashevsky, Shmuel and Pan, Michael and Bortnikov, Edward and Twitto, Moshe},
  journal={Proceedings of the VLDB Endowment},
  volume={15},
  number={11},
  pages={3071--3084},
  year={2022},
  publisher={VLDB Endowment}
}

@article{lu2017wisckey,
  title={Wisckey: Separating keys from values in ssd-conscious storage},
  author={Lu, Lanyue and Pillai, Thanumalayan Sankaranarayana and Gopalakrishnan, Hariharan and Arpaci-Dusseau, Andrea C and Arpaci-Dusseau, Remzi H},
  journal={ACM Transactions on Storage (TOS)},
  volume={13},
  number={1},
  pages={1--28},
  year={2017},
  publisher={ACM New York, NY, USA}
}

@inproceedings{balmau2017triad,
  title={TRIAD: Creating Synergies Between Memory, Disk and Log in Log Structured Key-Value Stores},
  author={Balmau, Oana and Didona, Diego and Guerraoui, Rachid and Zwaenepoel, Willy and Yuan, Huapeng and Arora, Aashray and Gupta, Karan and Konka, Pavan},
  booktitle={2017 USENIX Annual Technical Conference (USENIX ATC 17)},
  pages={363--375},
  year={2017}
}

@inproceedings{shetty2013building,
  title={Building Workload-Independent Storage with VT-Trees},
  author={Shetty, Pradeep J and Spillane, Richard P and Malpani, Ravikant R and Andrews, Binesh and Seyster, Justin and Zadok, Erez},
  booktitle={11th USENIX Conference on File and Storage Technologies (FAST 13)},
  pages={17--30},
  year={2013}
}

@inproceedings{thonangi2017log,
  title={On log-structured merge for solid-state drives},
  author={Thonangi, Risi and Yang, Jun},
  booktitle={2017 IEEE 33rd International Conference on Data Engineering (ICDE)},
  pages={683--694},
  year={2017},
  organization={IEEE}
}

@inproceedings{wang2024range,
  title={Range Cache: An Efficient Cache Component for Accelerating Range Queries on LSM-Based Key-Value Stores},
  author={Wang, Xiaoliang and Jin, Peiquan and Luo, Yongping and Chu, Zhaole},
  booktitle={2024 IEEE 40th International Conference on Data Engineering (ICDE)},
  pages={488--500},
  year={2024},
  organization={IEEE}
}

@inproceedings{wu2020ac,
  title={AC-Key: Adaptive caching for LSM-basedKey-Value stores},
  author={Wu, Fenggang and Yang, Ming-Hong and Zhang, Baoquan and Du, David HC},
  booktitle={2020 USENIX Annual Technical Conference (USENIX ATC 20)},
  pages={603--615},
  year={2020}
}

@inproceedings{sivasubramanian2012amazon,
  title={Amazon dynamoDB: a seamlessly scalable non-relational database service},
  author={Sivasubramanian, Swaminathan},
  booktitle={Proceedings of the 2012 ACM SIGMOD International Conference on Management of Data},
  pages={729--730},
  year={2012}
}

@inproceedings{hu2024aceso,
  title={Aceso: Achieving Efficient Fault Tolerance in Memory-Disaggregated Key-Value Stores},
  author={Hu, Zhisheng and Zuo, Pengfei and Chen, Yizou and Wang, Chao and Hu, Junliang and Yang, Ming-Chang},
  booktitle={Proceedings of the ACM SIGOPS 30th Symposium on Operating Systems Principles},
  pages={127--143},
  year={2024}
}

@article{bortnikov2018accordion,
  title={Accordion: Better memory organization for LSM key-value stores},
  author={Bortnikov, Edward and Braginsky, Anastasia and Hillel, Eshcar and Keidar, Idit and Sheffi, Gali},
  journal={Proceedings of the VLDB Endowment},
  volume={11},
  number={12},
  pages={1863--1875},
  year={2018},
  publisher={VLDB Endowment}
}

@inproceedings{golan2015scaling,
  title={Scaling concurrent log-structured data stores},
  author={Golan-Gueta, Guy and Bortnikov, Edward and Hillel, Eshcar and Keidar, Idit},
  booktitle={Proceedings of the Tenth European Conference on Computer Systems},
  pages={1--14},
  year={2015}
}

@inproceedings{balmau2019silk,
  title={SILK: Preventing Latency Spikes in Log-Structured Merge Key-Value Stores},
  author={Balmau, Oana and Dinu, Florin and Zwaenepoel, Willy and Gupta, Karan and Chandhiramoorthi, Ravishankar and Didona, Diego},
  booktitle={2019 USENIX Annual Technical Conference (USENIX ATC 19)},
  pages={753--766},
  year={2019}
}

@article{luo2019performance,
  title={On performance stability in LSM-based storage systems (extended version)},
  author={Luo, Chen and Carey, Michael J},
  journal={arXiv preprint arXiv:1906.09667},
  year={2019}
}

@article{ahmad2015compaction,
  title={Compaction management in distributed key-value datastores},
  author={Ahmad, Muhammad Yousuf and Kemme, Bettina},
  journal={Proceedings of the VLDB Endowment},
  volume={8},
  number={8},
  pages={850--861},
  year={2015},
  publisher={VLDB Endowment}
}

@inproceedings{vinccon2018noftl,
  title={Noftl-kv: Tackling write-amplification on kv-stores with native storage management},
  author={Vin{\c{c}}on, Tobias and Hardock, Sergej and Riegger, Christian and Oppermann, Julian and Koch, Andreas and Petrov, Ilia},
  booktitle={Advances in database technology-EDBT 2018: 21st International Conference on Extending Database Technology, Vienna, Austria, March 26-29, 2018. proceedings},
  pages={457--460},
  year={2018},
  organization={University of Konstanz, University Library}
}

@inproceedings{zhang2020fpga,
  title={FPGA-Accelerated Compactions for LSM-based Key-Value Store},
  author={Zhang, Teng and Wang, Jianying and Cheng, Xuntao and Xu, Hao and Yu, Nanlong and Huang, Gui and Zhang, Tieying and He, Dengcheng and Li, Feifei and Cao, Wei and others},
  booktitle={18th USENIX Conference on File and Storage Technologies (FAST 20)},
  pages={225--237},
  year={2020}
}

@inproceedings{wang2014efficient,
  title={An efficient design and implementation of LSM-tree based key-value store on open-channel SSD},
  author={Wang, Peng and Sun, Guangyu and Jiang, Song and Ouyang, Jian and Lin, Shiding and Zhang, Chen and Cong, Jason},
  booktitle={Proceedings of the Ninth European Conference on Computer Systems},
  pages={1--14},
  year={2014}
}

@inproceedings{sarkar2020lethe,
  title={Lethe: A tunable delete-aware LSM engine},
  author={Sarkar, Subhadeep and Papon, Tarikul Islam and Staratzis, Dimitris and Athanassoulis, Manos},
  booktitle={Proceedings of the 2020 ACM SIGMOD International Conference on Management of Data},
  pages={893--908},
  year={2020}
}

@article{benson2021viper,
  title={Viper: An efficient hybrid pmem-dram key-value store},
  author={Benson, Lawrence and Makait, Hendrik and Rabl, Tilmann},
  journal={Proceedings of the VLDB Endowment},
  volume={14},
  number={9},
  pages={1544--1556},
  year={2021},
  publisher={VLDB Endowment}
}

@article{kim2020robust,
  title={Robust and efficient memory management in Apache AsterixDB},
  author={Kim, Taewoo and Behm, Alexander and Blow, Michael and Borkar, Vinayak and Bu, Yingyi and Carey, Michael J and Hubail, Murtadha and Jahangiri, Shiva and Jia, Jianfeng and Li, Chen and others},
  journal={Software: Practice and Experience},
  volume={50},
  number={7},
  pages={1114--1151},
  year={2020},
  publisher={Wiley Online Library}
}

@inproceedings{knorr2022proteus,
  title={Proteus: A Self-Designing Range Filter},
  author={Knorr, Eric R and Lemaire, Baptiste and Lim, Andrew and Luo, Siqiang and Zhang, Huanchen and Idreos, Stratos and Mitzenmacher, Michael},
  booktitle={Proceedings of the 2022 International Conference on Management of Data},
  pages={1670--1684},
  year={2022}
}

@inproceedings{luo2020rosetta,
  title={Rosetta: A robust space-time optimized range filter for key-value stores},
  author={Luo, Siqiang and Chatterjee, Subarna and Ketsetsidis, Rafael and Dayan, Niv and Qin, Wilson and Idreos, Stratos},
  booktitle={Proceedings of the 2020 ACM SIGMOD International Conference on Management of Data},
  pages={2071--2086},
  year={2020}
}

@inproceedings{zhang2018surf,
  title={Surf: Practical range query filtering with fast succinct tries},
  author={Zhang, Huanchen and Lim, Hyeontaek and Leis, Viktor and Andersen, David G and Kaminsky, Michael and Keeton, Kimberly and Pavlo, Andrew},
  booktitle={Proceedings of the 2018 International Conference on Management of Data},
  pages={323--336},
  year={2018}
}

@article{sarkar2022constructing,
  title={Constructing and analyzing the LSM compaction design space},
  author={Sarkar, Subhadeep and Staratzis, Dimitris and Zhu, Zichen and Athanassoulis, Manos},
  journal={arXiv preprint arXiv:2202.04522},
  year={2022}
}

@inproceedings{absalyamov2018lightweight,
  title={Lightweight cardinality estimation in LSM-based systems},
  author={Absalyamov, Ildar and Carey, Michael J and Tsotras, Vassilis J},
  booktitle={Proceedings of the 2018 International Conference on Management of Data},
  pages={841--855},
  year={2018}
}

@article{ren2017slimdb,
  title={SlimDB: A space-efficient key-value storage engine for semi-sorted data},
  author={Ren, Kai and Zheng, Qing and Arulraj, Joy and Gibson, Garth},
  journal={Proceedings of the VLDB Endowment},
  volume={10},
  number={13},
  pages={2037--2048},
  year={2017},
  publisher={VLDB Endowment}
}

@article{yang2020leaper,
  title={Leaper: A learned prefetcher for cache invalidation in LSM-tree based storage engines},
  author={Yang, Lei and Wu, Hong and Zhang, Tieying and Cheng, Xuntao and Li, Feifei and Zou, Lei and Wang, Yujie and Chen, Rongyao and Wang, Jianying and Huang, Gui},
  journal={Proceedings of the VLDB Endowment},
  volume={13},
  number={12},
  pages={1976--1989},
  year={2020},
  publisher={VLDB Endowment}
}

@inproceedings{dayan2021chucky,
  title={Chucky: A Succinct Cuckoo Filter for LSM-Tree},
  author={Dayan, Niv and Twitto, Moshe},
  booktitle={Proceedings of the 2021 International Conference on Management of Data},
  pages={365--378},
  year={2021}
}

@article{yang2022oceanbase,
  title={OceanBase: a 707 million tpmC distributed relational database system},
  author={Yang, Zhenkun and Yang, Chuanhui and Han, Fusheng and Zhuang, Mingqiang and Yang, Bing and Yang, Zhifeng and Cheng, Xiaojun and Zhao, Yuzhong and Shi, Wenhui and Xi, Huafeng and others},
  journal={Proceedings of the VLDB Endowment},
  volume={15},
  number={12},
  pages={3385--3397},
  year={2022},
  publisher={VLDB Endowment}
}

@misc{elasticache,
  title        = {Amazon ElastiCache},
  author       = {Amazon Web Services},
  year         = {2024},
  howpublished         = {\url{https://aws.amazon.com/elasticache/}},
}

@misc{workers_kv,
  title        = {Cloudflare Workers KV},
  author       = {{Cloudflare, Inc.}},
  year         = {2024},
  howpublished = {\url{https://developers.cloudflare.com/workers/runtime-apis/kv/}},
}

@misc{memorystore,
  title        = {Google Cloud Memorystore},
  author       = {{Google Cloud}},
  year         = {2024},
  howpublished = {\url{https://cloud.google.com/memorystore}}
}

@inproceedings{zhang2018elasticbf,
  title={ElasticBF: Fine-grained and Elastic Bloom Filter Towards Efficient Read for LSM-tree-based KV Stores},
  author={Zhang, Yueming and Li, Yongkun and Guo, Fan and Li, Cheng and Xu, Yinlong},
  booktitle={10th USENIX Workshop on Hot Topics in Storage and File Systems (HotStorage 18)},
  year={2018}
}

@inproceedings{curino2011workload,
  title={Workload-aware database monitoring and consolidation},
  author={Curino, Carlo and Jones, Evan PC and Madden, Samuel and Balakrishnan, Hari},
  booktitle={Proceedings of the 2011 ACM SIGMOD International Conference on Management of data},
  pages={313--324},
  year={2011}
}

@inproceedings{qader2018comparative,
  title={A comparative study of secondary indexing techniques in LSM-based NoSQL databases},
  author={Qader, Mohiuddin Abdul and Cheng, Shiwen and Hristidis, Vagelis},
  booktitle={Proceedings of the 2018 International Conference on Management of Data},
  pages={551--566},
  year={2018}
}

@inproceedings{yao2017light,
  title={A light-weight compaction tree to reduce I/O amplification toward efficient key-value stores},
  author={Yao, Ting and Wan, Jiguang and Huang, Ping and He, Xubin and Gui, Qingxin and Wu, Fei and Xie, Changsheng},
  booktitle={Proc. 33rd Int. Conf. Massive Storage Syst. Technol.(MSST)},
  pages={1--13},
  year={2017}
}

@inproceedings{Amdahllaw,
author = {Amdahl, Gene M.},
title = {Validity of the single processor approach to achieving large scale computing capabilities},
year = {1967},
isbn = {9781450378956},
publisher = {Association for Computing Machinery},
address = {New York, NY, USA},
url = {https://doi.org/10.1145/1465482.1465560},
doi = {10.1145/1465482.1465560},
pages = {483–485},
numpages = {3},
location = {Atlantic City, New Jersey},
series = {AFIPS '67 (Spring)}
}

@misc{flinkwindows,
  title={Apache Flink Documentation: Windows},
  author={Apache Flink},
  howpublished={\url{https://nightlies.apache.org/flink/flink-docs-stable/docs/dev/datastream/operators/windows/}},
  year={2025}
}

@misc{Wiredtigerdoc,
  title={WiredTiger API: WT\_SESSION Struct Reference},
  author={MongoDB},
  howpublished={\url{http://source.wiredtiger.com/mongodb-5.0/struct_w_t___s_e_s_s_i_o_n.html}},
  year={2025}
}

@article{yao2017building,
  title={Building efficient key-value stores via a lightweight compaction tree},
  author={Yao, Ting and Wan, Jiguang and Huang, Ping and He, Xubin and Wu, Fei and Xie, Changsheng},
  journal={ACM Transactions on Storage (TOS)},
  volume={13},
  number={4},
  pages={1--28},
  year={2017},
  publisher={ACM New York, NY, USA}
}

@inproceedings{gmach2007workload,
  title={Workload analysis and demand prediction of enterprise data center applications},
  author={Gmach, Daniel and Rolia, Jerry and Cherkasova, Ludmila and Kemper, Alfons},
  booktitle={2007 IEEE 10th International Symposium on Workload Characterization},
  pages={171--180},
  year={2007},
  organization={IEEE}
}

@article{yu2022treeline,
  title={{TreeLine: an update-in-place key-value store for modern storage}},
  author={{Yu, Geoffrey X and Markakis, Markos and Kipf, Andreas and Larson, Per-{\AA}ke and Minhas, Umar Farooq and Kraska, Tim}},
  journal={Proceedings of the VLDB Endowment},
  volume={16},
  number={1},
  pages={99--112},
  year={2022},
  publisher={VLDB Endowment}
}

@inproceedings{zhang2023pm,
  title={PM-Blade: A Persistent Memory Augmented LSM-tree Storage for Database},
  author={Zhang, Yinan and Hu, Huiqi and Zhou, Xuan and Xie, Enlong and Ren, Hongdi and Jin, Le},
  booktitle={2023 IEEE 39th International Conference on Data Engineering (ICDE)},
  pages={3363--3375},
  year={2023},
  organization={IEEE}
}

@inproceedings{xiong2024distore,
  title={DiStore: A Fully Memory Disaggregation Friendly Key-Value Store with Improved Tail Latency and Space Efficiency},
  author={Xiong, Ziwei and Jiang, Dejun and Xiong, Jin},
  booktitle={Proceedings of the 53rd International Conference on Parallel Processing},
  pages={607--617},
  year={2024}
}

@inproceedings{zhu2024memory,
  title={In-MemoryKey-Value Store Live Migration with NetMigrate},
  author={Zhu, Zeying and Zhao, Yibo and Liu, Zaoxing},
  booktitle={22nd USENIX Conference on File and Storage Technologies (FAST 24)},
  pages={209--224},
  year={2024}
}

@inproceedings{mammarella2009modular,
  title={Modular data storage with Anvil},
  author={Mammarella, Mike and Hovsepian, Shant and Kohler, Eddie},
  booktitle={Proceedings of the ACM SIGOPS 22nd symposium on Operating systems principles},
  pages={147--160},
  year={2009}
}

@inproceedings{kaiyrakhmet2019slm,
  title={SLM-DB:Single-LevelKey-Value store with persistent memory},
  author={Kaiyrakhmet, Olzhas and Lee, Songyi and Nam, Beomseok and Noh, Sam H and Choi, Young-ri},
  booktitle={17th USENIX Conference on File and Storage Technologies (FAST 19)},
  pages={191--205},
  year={2019}
}

@inproceedings{kejriwal2016slik,
  title={SLIK: Scalable Low-Latency Indexes for a Key-Value Store},
  author={Kejriwal, Ankita and Gopalan, Arjun and Gupta, Ashish and Jia, Zhihao and Yang, Stephen and Ousterhout, John},
  booktitle={2016 USENIX Annual Technical Conference (USENIX ATC 16)},
  pages={57--70},
  year={2016}
}

@inproceedings{shen2023fusee,
  title={FUSEE: A fully Memory-DisaggregatedKey-Value store},
  author={Shen, Jiacheng and Zuo, Pengfei and Luo, Xuchuan and Yang, Tianyi and Su, Yuxin and Zhou, Yangfan and Lyu, Michael R},
  booktitle={21st USENIX Conference on File and Storage Technologies (FAST 23)},
  pages={81--98},
  year={2023}
}

@article{camal,
author = {Yu, Weiping and Luo, Siqiang and Yu, Zihao and Cong, Gao},
title = {CAMAL: Optimizing LSM-trees via Active Learning},
year = {2024},
issue_date = {September 2024},
publisher = {Association for Computing Machinery},
address = {New York, NY, USA},
volume = {2},
number = {4},
}

@misc{adya2019fast,
  title={Fast key-value stores: An idea whose time has come and gone (HotOS’19 talk slides)},
  author={Adya, Atul and Myers, Daniel and Qin, Henry and Grandl, Robert},
  year={2019}
}

@inproceedings{zhang2022depart,
  title={DEPART: Replica Decoupling for Distributed Key-Value Storage},
  author={Zhang, Qiang and Li, Yongkun and Lee, Patrick PC and Xu, Yinlong and Wu, Si},
  booktitle={20th USENIX Conference on File and Storage Technologies (FAST 22)},
  pages={397--412},
  year={2022}
}

@article{wei2021xstore,
  title={Xstore: Fast rdma-based ordered key-value store using remote learned cache},
  author={Wei, Xingda and Chen, Rong and Chen, Haibo and Zang, Binyu},
  journal={ACM Transactions on Storage (TOS)},
  volume={17},
  number={3},
  pages={1--32},
  year={2021},
  publisher={ACM New York, NY}
}

@inproceedings{wang2023replicating,
  title={Replicating Persistent Memory Key-Value Stores with Efficient RDMA Abstraction},
  author={Wang, Qing and Lu, Youyou and Wang, Jing and Shu, Jiwu},
  booktitle={17th USENIX Symposium on Operating Systems Design and Implementation (OSDI 23)},
  pages={441--459},
  year={2023}
}

@inproceedings{zhong2023redesigning,
  title={Redesigning High-Performance LSM-based Key-Value Stores with Persistent CPU Caches},
  author={Zhong, Yijie and Shen, Zhirong and Yu, Zixiang and Shu, Jiwu},
  booktitle={2023 IEEE 39th International Conference on Data Engineering (ICDE)},
  pages={1098--1111},
  year={2023},
  organization={IEEE}
}

@article{chang2008bigtable,
  title={Bigtable: A distributed storage system for structured data},
  author={Chang, Fay and Dean, Jeffrey and Ghemawat, Sanjay and Hsieh, Wilson C and Wallach, Deborah A and Burrows, Mike and Chandra, Tushar and Fikes, Andrew and Gruber, Robert E},
  journal={ACM Transactions on Computer Systems (TOCS)},
  volume={26},
  number={2},
  pages={1--26},
  year={2008},
  publisher={ACM New York, NY, USA}
}

@article{chen2024oasis,
  title={Oasis: An Optimal Disjoint Segmented Learned Range Filter},
  author={Chen, Guanduo and He, Zhenying and Li, Meng and Luo, Siqiang},
  journal={Proceedings of the VLDB Endowment},
  volume={17},
  number={8},
  pages={1911--1924},
  year={2024},
  publisher={VLDB Endowment}
}

@inproceedings{li2022seesaw,
  title={Seesaw Counting Filter: An Efficient Guardian for Vulnerable Negative Keys During Dynamic Filtering},
  author={Li, Meng and Chen, Deyi and Dai, Haipeng and Xie, Rongbiao and Luo, Siqiang and Gu, Rong and Yang, Tong and Chen, Guihai},
  booktitle={Proceedings of the ACM Web Conference 2022},
  pages={2759--2767},
  year={2022}
}

@inproceedings{facebookworkload,
author = {Atikoglu, Berk and Xu, Yuehai and Frachtenberg, Eitan and Jiang, Song and Paleczny, Mike},
title = {Workload analysis of a large-scale key-value store},
year = {2012},
isbn = {9781450310970},
publisher = {Association for Computing Machinery},
address = {New York, NY, USA},
url = {https://doi-org.remotexs.ntu.edu.sg/10.1145/2254756.2254766},
doi = {10.1145/2254756.2254766},
booktitle = {Proceedings of the 12th ACM SIGMETRICS/PERFORMANCE Joint International Conference on Measurement and Modeling of Computer Systems},
pages = {53–64},
numpages = {12},
location = {London, England, UK},
series = {SIGMETRICS '12}
}

@book{partitionprob,
  author    = {Michael R. Garey and David S. Johnson},
  title     = {Computers and Intractability: A Guide to the Theory of NP-Completeness},
  year      = {1979},
  publisher = {W. H. Freeman},
  pages     = {96--105},
  isbn      = {978-0-7167-1045-5}
}

\end{document}